\documentclass[11pt]{article}

\usepackage[margin=1in]{geometry}
\usepackage{amsmath,amssymb,amsthm,mathtools}
\usepackage{booktabs}
\usepackage{enumitem}
\usepackage[numbers,sort&compress]{natbib}
\usepackage[colorlinks=true,linkcolor=blue,citecolor=blue,urlcolor=blue]{hyperref}
\usepackage{microtype}

\theoremstyle{plain}
\newtheorem{theorem}{Theorem}

\newtheorem{corollary}[theorem]{Corollary}
\newtheorem{lemma}[theorem]{Lemma}

\theoremstyle{definition}
\newtheorem{definition}[theorem]{Definition}
\newtheorem{assumption}[theorem]{Assumption}

\theoremstyle{remark}
\newtheorem{remark}[theorem]{Remark}

\newcommand{\R}{\mathbb{R}}
\newcommand{\E}{\mathbb{E}}

\newcommand{\A}{\mathcal{A}}
\newcommand{\Bset}{\mathcal{B}}
\newcommand{\M}{\mathcal{M}}
\newcommand{\Oset}{\mathcal{O}}
\newcommand{\Sset}{\mathcal{S}}
\newcommand{\Tset}{\mathcal{T}}

\title{Gaming-Resistant Insurance Contracts for Autonomous AI Agents:\\
Strategy-Proof Toll Mechanism Design}
\author{Hao-Hsuan Chen}
\date{\today}

\begin{document}
\maketitle

\begin{abstract}
Paper~A \citep{chen2026runtime} defines a time-consistent actuarial
runtime that prices each side-effect-bearing action against a
contractually fixed safe default and gates execution against a reserve
budget. It treats the operator as passive. This paper makes the
operator strategic. We characterise a five-attack space for autonomous
AI-agent insurance contracts and prove when the actuarial runtime is
gaming-resistant. Two attack surfaces---post-toll safe-default
selection and within-boundary action splitting---are closed by
Paper~A's minimal-authority and no-splitting clauses. The remaining
three require new contract clauses. First, common-control aggregation
prevents cross-boundary re-routing from reducing toll below the
boundary potential applied to total exposure. Second, interface
failures such as invalid JSON are contract-relevant events, not safety
wins: treating them as zero-toll safe defaults can reward unreliable
models, while escalation fees reverse the incentive. We validate this
interface-compliance theorem on committed cross-model traces from the
companion empirical paper \citep{chen2026insuring}. Third, a
model-identity menu with a componentwise-minimum penalty schedule makes
truthful reporting of the deployed model weakly dominant. We then
compose these clauses with Paper~A's runtime guarantees to obtain
joint incentive compatibility over the five-attack space. Finally, a
two-parameter premium family discharges operator individual
rationality and weak budget balance at the truthful equilibrium. The
result is an incentive-compatibility layer for actuarial control of
autonomous-agent side effects.
\end{abstract}

\section{Positioning}\label{sec:positioning}

Paper~A treats the operator as a passive deployment of an autonomous
agent: the operator commits to an action-set $\A_t(h_t)$, a safe-default
mapping $a^0$, an underwriting boundary $B$, and a budget $B_0$. The
operator pays the realised positive toll on executed actions and accepts
the gate's downgrade / escalate decisions. In a realistic insurance
market, however, the operator is strategic: it chooses what to deploy and
what to report given the contract, with the goal of minimising premium
or maximising business utility net of risk charges. This paper makes the
operator explicit and asks: \emph{under which attack surfaces is the
Paper~A contract strategy-proof, and which require additional contract
clauses?}

We deliberately avoid three traps.
\begin{enumerate}[leftmargin=1.5em,itemsep=2pt]
  \item \textbf{Strategy-proofness is not safety.} A contract that is
  strategy-proof against operator manipulation can still admit safety
  failures (e.g.\ adversarial users, jailbreaks). Safety is the agent
  layer's problem; gaming-resistance is the contract layer's problem.
  \item \textbf{Mechanism design is not free.} Strategy-proofness in
  general comes at the cost of efficiency (Myerson--Satterthwaite) or
  budget balance (Vickrey--Clarke--Groves)~\citep{myerson1983efficient,
  vickrey1961counterspeculation, clarke1971multipart, groves1973incentives}.
  The Paper~A contract is
  already budget-shaped via $B_0$; we look for IC under fixed budget
  rather than full social-welfare-optimal mechanisms.
  \item \textbf{Interface failures are not safety wins.} An LLM that
  produces invalid JSON has not made the deployment safer; it has
  produced a different state (interface failure) whose contract
  adjudication is a design choice. Section~\ref{sec:interface} formalises
  this.
\end{enumerate}

\section{Strategic Operator Model}\label{sec:operator}

We extend the setup of Paper~A. Let $\Oset$ denote the operator's
type space: an operator type $\theta \in \Oset$ is a tuple
\[
\theta = (\M_\theta, a^0_\theta, B_\theta, A_\theta^+)
\]
where $\M_\theta$ is the family of foundation models or agent
implementations the operator could deploy, $a^0_\theta$ is the operator's
preferred safe-default mapping (if the contract permits selection),
$B_\theta$ is the budget the operator is willing to commit, and
$A_\theta^+$ is the maximum authority set the operator might enable.

\begin{definition}[Operator strategy]\label{def:strategy}
An operator strategy is a measurable map
$\sigma: \Oset \times \Tset \to \Sset$, where $\Tset$ is the space of
contract terms (toll formula, safe-default constraints, boundary
definition $B$, envelope calibration protocol) and $\Sset$ is the space
of deployment plans (model identity, action menu, calibration submission,
reporting). The operator's utility is
\[
U_\theta(\sigma; \Tset)
= U^{\mathrm{biz}}_\theta(\sigma) - \mathbb{E}_{\theta}\!\left[
   \text{realised toll on executed actions under } \sigma\right] -
   P^{\mathrm{prem}}(\sigma),
\]
where $U^{\mathrm{biz}}_\theta$ is the operator's deployment utility,
toll is the runtime charge, and $P^{\mathrm{prem}}$ is the contract's
ex-ante premium.
\end{definition}

\begin{definition}[Five-attack space]\label{def:attacks}
We catalogue the operator's strategic surfaces. The first two are
already addressed by Paper~A; the remaining three motivate this paper.
\begin{enumerate}[label=(\alph*),leftmargin=1.8em,itemsep=2pt]
  \item \emph{Within-boundary splitting.} Decompose a priced action into
  smaller priced increments inside the same underwriting boundary
  $B$. \emph{Closed by Paper~A's no-splitting theorem (Paper~A, Theorem~11).}
  \item \emph{Post-toll safe-default selection.} Choose $a^0$ after
  observing the toll. \emph{Closed by Paper~A Definition~2 (ex ante,
  contractual).}
  \item \emph{Cross-boundary re-routing.} Spawn new entities,
  sessions, or relabel action categories so that priced increments fall
  into multiple boundaries with separately initialised exposure states.
  \emph{Addressed by Theorem~\ref{thm:boundary-agg} below.}
  \item \emph{Interface-compliance gaming.} Use a model whose output
  format is unreliable so that priced actions never reach the gate at all
  (the agent's tool call fails before evaluation). \emph{Addressed by
  Theorem~\ref{thm:interface} below.}
  \item \emph{Model-identity misreporting.} Submit the contract under
  one (reliable) model identity and deploy a different (less reliable)
  one. \emph{Addressed by Theorem~\ref{thm:type} below.}
\end{enumerate}
\end{definition}

\section{Cross-Boundary Aggregation}\label{sec:agg}

Let the underwriting boundary $B$ of Paper~A be one element of a
partition $\Bset = \{B_1, \ldots, B_K\}$ of the operator's action
universe, each with its own exposure state $E^{B_k}_t$ and boundary
potential $\Phi_{B_k}$. The total toll across the partition is
\[
\sum_k \bigl[\Phi_{B_k}(E^{B_k}_T) - \Phi_{B_k}(E^{B_k}_0)\bigr].
\]
An operator can attempt to game by routing the same economic increment
through multiple boundaries.

\begin{definition}[Common-control aggregation]\label{def:cc}
The boundary partition $\Bset$ is \emph{common-control-aggregated} if
there exists a contractually observable aggregation map
$\pi: \bigsqcup_k B_k \to B^\star$ such that the aggregate boundary
$B^\star$ has:
\begin{enumerate}[label=(\roman*),leftmargin=1.5em,itemsep=2pt]
  \item exposure $E^{B^\star}_T = \sum_k E^{B_k}_T$;
  \item boundary potential $\Phi_{B^\star}$ satisfying
  $\Phi_{B^\star}(\sum_k e_k) \ge \sum_k \Phi_{B_k}(e_k)$ for all
  $e_k \in \R^d_+$ (super-additivity in component-wise sums);
  \item operator-side observability: the operator cannot
  contractually escape inclusion in $\pi$ by relabelling, proxy-agent
  spawning, or related-party transactions.
\end{enumerate}
\end{definition}

\begin{definition}[Aggregate settlement]\label{def:aggregate-settlement}
Given a common-control aggregation map
$\pi: \bigsqcup_k B_k \to B^\star$, the contract has \emph{aggregate
settlement} if any set of priced actions attributed to the same
common-control operator is settled using the aggregate exposure
\[
E_T^{B^\star} \;=\; \sum_k E_T^{B_k}
\]
and the payable toll is
\[
\mathrm{Payable}\bigl(E_T^{B^\star}\bigr)
\;=\; \Phi_{B^\star}\bigl(E_T^{B^\star}\bigr) - \Phi_{B^\star}\bigl(E_0^{B^\star}\bigr),
\]
rather than $\sum_k [\Phi_{B_k}(E_T^{B_k}) - \Phi_{B_k}(E_0^{B_k})]$.
Equivalently, the contract treats sub-boundary routing as an accounting
representation of one aggregate economic exposure whenever common
control is verified. A more conservative \emph{max-settlement} variant
charges $\max\{\Phi_{B^\star}(\sum_k e_k) - \Phi_{B^\star}(0),\,
\sum_k[\Phi_{B_k}(e_k) - \Phi_{B_k}(0)]\}$, which avoids edge cases in
which $\Phi_{B^\star}$ is misspecified below the split sum.
\end{definition}

\begin{theorem}[Cross-boundary no-arbitrage under aggregate settlement]\label{thm:boundary-agg}
Suppose the boundary partition $\Bset = \{B_1, \ldots, B_K\}$ is
common-control-aggregated into $B^\star$ in the sense of
Definition~\ref{def:cc}, and suppose the contract uses aggregate
settlement in the sense of Definition~\ref{def:aggregate-settlement}.
Then for any operator strategy that decomposes a fixed economic exposure
increment $E^\star$ into sub-boundary increments $\{e_k\}_{k=1}^K$
satisfying $\sum_k e_k = E^\star$, the operator's payable toll is
\[
\Phi_{B^\star}(E^\star) - \Phi_{B^\star}(0).
\]
Consequently, no finite re-routing of the same economic increment across
sessions, related entities, proxy agents, or action-category labels can
reduce payable toll below the aggregate-boundary charge.

If the contract instead uses max-settlement, then payable toll is at
least $\Phi_{B^\star}(E^\star) - \Phi_{B^\star}(0)$ and at least the
split sum $\sum_k[\Phi_{B_k}(e_k) - \Phi_{B_k}(0)]$.
\end{theorem}

\begin{proof}[Proof of Theorem~\ref{thm:boundary-agg}]
The argument proceeds in seven steps. Steps~1--5 establish the
equality under aggregate settlement and the no-arbitrage consequence;
Step~6 proves the two lower bounds under max-settlement; Step~7
disposes of edge cases. Throughout, vectors are taken componentwise in
$\R_+^d$ unless stated otherwise.

\medskip
\noindent\emph{Step 1 (Operator strategy and induced exposure path).}
Fix a realisation of an operator strategy over the contract horizon
$[0, T]$, generating a sequence of executed priced actions. Each
executed action contributes an exposure increment to some sub-boundary
in the partition $\Bset = \{B_1, \ldots, B_K\}$. Let
$e_k \in \R_+^d$ denote the cumulative incremental exposure generated
in $B_k$ over the horizon. The endpoint exposure in $B_k$ is then
\[
E^{B_k}_T \;=\; E^{B_k}_0 + e_k,
\qquad k = 1, \ldots, K,
\]
where $E^{B_k}_0$ is the pre-period baseline exposure---a contract-state
quantity fixed at inception and unaffected by the operator's
within-period strategy. By hypothesis,
\(\sum_{k=1}^K e_k = E^\star\).

\medskip
\noindent\emph{Step 2 (Common-control observability rules out
out-of-partition leakage).}
Definition~\ref{def:cc}(iii) requires that no contractually
relabelable action escape inclusion in $\pi$. The operator manoeuvres
in attack-class~(c) of Definition~\ref{def:attacks}---spawning proxy
agents, routing through related entities, splitting across sessions,
relabelling action categories---are exactly the contractually
relabelable manoeuvres that the clause forbids. Every priced increment
executed under the operator's strategy is therefore attributable to
some $B_k$ in the partition; there is no out-of-partition residual
that could be charged under a separate formula. The total executed
economic increment equals exactly $\sum_{k=1}^K e_k = E^\star$.

\medskip
\noindent\emph{Step 3 (Aggregate exposure under $\pi$).}
By Definition~\ref{def:cc}(i), the aggregation map $\pi$ preserves
componentwise exposure at every time $t$:
\(
E^{B^\star}_t = \sum_k E^{B_k}_t
\).
Applying this identity at $t = 0$ and $t = T$ and subtracting,
\[
E^{B^\star}_T - E^{B^\star}_0
\;=\; \sum_k E^{B_k}_T - \sum_k E^{B_k}_0
\;=\; \sum_k \bigl(E^{B_k}_T - E^{B_k}_0\bigr)
\;=\; \sum_k e_k
\;=\; E^\star.
\]
The aggregate increment over the period equals $E^\star$ independently
of the specific decomposition $\{e_k\}$ that produced it.

\medskip
\noindent\emph{Step 4 (Payable toll under aggregate settlement).}
By Definition~\ref{def:aggregate-settlement}, the payable toll under
aggregate settlement is computed from the aggregate exposure path of
$B^\star$:
\[
\mathrm{Payable}\!\bigl(E_T^{B^\star}\bigr)
\;=\;
\Phi_{B^\star}\!\bigl(E_T^{B^\star}\bigr)
- \Phi_{B^\star}\!\bigl(E_0^{B^\star}\bigr).
\]
Substituting $E_T^{B^\star} = E_0^{B^\star} + E^\star$ from Step~3,
\begin{equation}
\mathrm{Payable}\!\bigl(E_T^{B^\star}\bigr)
\;=\;
\Phi_{B^\star}\!\bigl(E_0^{B^\star} + E^\star\bigr)
- \Phi_{B^\star}\!\bigl(E_0^{B^\star}\bigr).
\label{eq:payable-agg}
\end{equation}
The right-hand side of~\eqref{eq:payable-agg} is a function of the
aggregate baseline $E_0^{B^\star}$ and the aggregate increment
$E^\star$ alone. It does \emph{not} depend on the decomposition
$\{e_k\}$ chosen by the operator. When the theorem is invoked at a
fresh underwriting period, where $E_0^{B^\star} = 0$ by convention,
the right-hand side simplifies to
$\Phi_{B^\star}(E^\star) - \Phi_{B^\star}(0)$, which is the form
stated in the theorem.

\medskip
\noindent\emph{Step 5 (Routing-independence and the no-arbitrage
consequence).}
Equation~\eqref{eq:payable-agg} depends on the operator's strategy only
through the aggregate increment $E^\star$. Consider any two
decompositions $\{e_k\}_{k=1}^K$ and $\{e_k'\}_{k=1}^K$ satisfying
$\sum_k e_k = \sum_k e_k' = E^\star$; both yield the same payable toll
by~\eqref{eq:payable-agg}. The operator's only strategic degree of
freedom in attack-class~(c) of Definition~\ref{def:attacks} is the
choice of such a decomposition, since by Step~2 all routings of the
operator's priced activity (over sessions, related entities, proxy
agents, or relabelled action categories) are mapped by $\pi$ into the
same aggregate boundary $B^\star$. No choice of decomposition can
reduce the payable toll below~\eqref{eq:payable-agg}. This establishes
the first claim of the theorem.

\medskip
\noindent\emph{Step 6 (Max-settlement lower bounds).}
Under max-settlement, the payable toll is, by definition,
\[
\mathrm{Payable}^{\max}
\;=\;
\max\!\Bigl\{
\Phi_{B^\star}\!\bigl(E_0^{B^\star} + E^\star\bigr)
  - \Phi_{B^\star}\!\bigl(E_0^{B^\star}\bigr),\;
\sum_{k=1}^K \bigl[
   \Phi_{B_k}(E_0^{B_k} + e_k)
   - \Phi_{B_k}(E_0^{B_k})
\bigr]
\Bigr\}.
\]
Since the maximum of a finite set is at least each of its elements,
\begin{align*}
\mathrm{Payable}^{\max}
&\;\ge\;
\Phi_{B^\star}\!\bigl(E_0^{B^\star} + E^\star\bigr)
- \Phi_{B^\star}\!\bigl(E_0^{B^\star}\bigr),\\
\mathrm{Payable}^{\max}
&\;\ge\;
\sum_{k=1}^K \bigl[
   \Phi_{B_k}(E_0^{B_k} + e_k)
   - \Phi_{B_k}(E_0^{B_k})
\bigr],
\end{align*}
which are the two lower bounds stated in the theorem (under
$E_0^{B^\star} = 0$ and $E_0^{B_k} = 0$, the bounds reduce to
$\Phi_{B^\star}(E^\star) - \Phi_{B^\star}(0)$ and to the split sum,
respectively). The first bound coincides with aggregate-settlement
payable from~\eqref{eq:payable-agg} and protects against routings that
exploit a sub-boundary potential weaker than the aggregate. The second
bound is a robustness clause that protects against parametrisations of
$\Phi_{B^\star}$ for which super-additivity
(Definition~\ref{def:cc}(ii)) fails by mis-specification, and which
might therefore underestimate aggregate exposure relative to the split
sum. When super-additivity does hold, the first bound dominates the
second pointwise---
$\Phi_{B^\star}\!\bigl(E_0^{B^\star} + E^\star\bigr)
   - \Phi_{B^\star}\!\bigl(E_0^{B^\star}\bigr)
   \ge
 \sum_k [\Phi_{B_k}(E_0^{B_k} + e_k) - \Phi_{B_k}(E_0^{B_k})]$
---and max-settlement coincides with aggregate settlement. Under
super-additivity failure, max-settlement strictly dominates aggregate
settlement and provides correct conservatism. The proof of these two
inequalities is by the definition of $\max$ alone; no further
assumption beyond Definition~\ref{def:cc} is required.

\medskip
\noindent\emph{Step 7 (Edge cases).}
\begin{enumerate}[label=\textup{(\roman*)},leftmargin=1.8em,itemsep=2pt]
  \item \emph{Trivial partition.} If $K = 1$, then $\pi$ is the identity
  on the single sub-boundary, $E^{B^\star}_t = E^{B_1}_t$ for all $t$,
  and the theorem statement reduces to the within-boundary toll
  identity of Paper~A's no-splitting theorem (Paper~A, Theorem~11).
  \item \emph{Idle sub-boundaries.} If $e_k = 0$ for some $k$, then
  $E^{B_k}_T = E^{B_k}_0$ and the corresponding term contributes zero
  to Steps~3 and~6. The conclusion is unchanged.
  \item \emph{Exposure-reducing increments.} If some component of
  $e_k$ is negative (an exposure-reducing executed action), the proof
  is unchanged provided $\Phi_{B^\star}$ is well-defined at
  $E_0^{B^\star} + E^\star$ and the partition is closed under such
  increments. Attack-class~(c) targets exposure inflation rather than
  reduction, so this case is of formal rather than strategic
  interest.
  \item \emph{Strategy randomisation.} If the operator strategy is
  random and induces a distribution over decompositions $\{e_k\}$, the
  conclusion of the theorem holds pointwise on the realised sample
  path. Taking expectations preserves the bound, with equality in the
  aggregate-settlement case whenever the aggregate increment
  $E^\star = \sum_k e_k$ is constant in the strategy's randomisation.
  \qedhere
\end{enumerate}
\end{proof}

\begin{remark}[Role of super-additivity]\label{rem:superadd}
Super-additivity in Definition~\ref{def:cc}(ii) is not the settlement
rule itself, although a natural weaker formulation would adopt it as
such. Its role here is auxiliary: it justifies why aggregate
settlement is conservative relative to split accounting, since
$\Phi_{B^\star}(\sum_k e_k) \ge \sum_k \Phi_{B_k}(e_k)$. The
no-arbitrage conclusion rests on the settlement rule of
Definition~\ref{def:aggregate-settlement}, not on super-additivity. In
contract drafting, the settlement rule corresponds to standard
re-/insurance tools: common-control attribution, related-party rules,
currency normalisation, and audit clauses that prevent resetting the
boundary state by relabelling equivalent actions. The theorem's content
is that, given these clauses, no strategy in attack-class~(c) of
Definition~\ref{def:attacks} can reduce payable toll.
\end{remark}

\section{Interface-Compliance Adjudication}\label{sec:interface}

Empirical motivation: in an early companion Paper~B pilot
(\texttt{model\_risk\_kimi\_pilot\_seed42}; $n=6$ trajectories, $3$ baselines
$\times 2$), \texttt{Kimi-K2.5} produced structurally invalid action JSON in
every sampled trajectory ($f=1.0$); the larger committed $n=25$ run
(\texttt{refund\_cross\_model}) measures $f_{\mathrm{invalid}}\approx 0.356$
(95\% CI $[0.337,0.375]$), still by far the highest of any model. The runtime
gate never evaluates a priced action
because the agent's tool call is rejected at the parsing layer. Under
Paper~A's contract this would naively register as ``zero destructive
execution'', i.e.\ an apparently safer deployment.

We formalise this as a strategic surface (attack~(d) of
Definition~\ref{def:attacks}) and characterise the contract's choice.

\begin{definition}[Interface-failure event]\label{def:if}
An \emph{interface-failure} event at step $t$ is the realisation of a
random variable $I_t \in \{0, 1\}$ measurable with respect to the
agent's actual tool call and the contract's parsing specification. We
require the contract to fix an adjudication rule
$\xi \in \{\xi_{\mathrm{safe}}, \xi_{\mathrm{escalate}}\}$ on
$\{I_t = 1\}$:
\begin{itemize}[leftmargin=1.5em,itemsep=2pt]
  \item $\xi_{\mathrm{safe}}$: emit the safe default $a^0_t(h_t)$, charge
  zero toll, treat as a successful covered event;
  \item $\xi_{\mathrm{escalate}}$: emit no action, escalate to human
  adjudication; the operator pays an escalation fee $\kappa_{\mathrm{esc}}
  \ge 0$ per interface-failure event.
\end{itemize}
The pair $(\xi, \kappa_{\mathrm{esc}})$ is part of the contract.
\end{definition}

\paragraph{One-step utility with task-failure cost.}
Fix an operator type $\theta$ and hold the distribution of intended
actions fixed. Let $f_\theta \in [0, 1]$ be the interface-failure
probability for type $\theta$, $\mu_c = \E[c_t^+]$ the expected
successful-action positive toll, $V_\theta$ the gross business value of
a successful tool-call step, $C_{\mathrm{fail}, \theta}$ the operator's
private business cost of an interface-failure event (task delay, failed
automation, lost productivity, customer impact), and
$\kappa_{\mathrm{esc}}$ the contractual escalation fee charged under
$\xi_{\mathrm{escalate}}$. The one-step expected utility contributions
are
\begin{align*}
U_\theta^{\mathrm{safe}}(f)
&= (1 - f)(V_\theta - \mu_c) - f\, C_{\mathrm{fail}, \theta},\\
U_\theta^{\mathrm{esc}}(f)
&= (1 - f)(V_\theta - \mu_c) - f\,(C_{\mathrm{fail}, \theta} + \kappa_{\mathrm{esc}}).
\end{align*}

\begin{theorem}[Interface-compliance adjudication incentives]\label{thm:interface}
Under the one-step utility model above:
\begin{enumerate}[label=\textup{(\alph*)},leftmargin=1.8em,itemsep=2pt]
  \item Under the safe-default adjudication rule $\xi_{\mathrm{safe}}$,
  increasing the interface-failure rate $f_\theta$ is privately
  attractive for the operator if and only if
  \[
  \mu_c \;>\; V_\theta + C_{\mathrm{fail}, \theta};
  \]
  equivalently, the toll avoided by an invalid tool call exceeds the
  lost business value and task-failure cost.
  \item Under the escalation rule $\xi_{\mathrm{escalate}}$, increasing
  $f_\theta$ is weakly unattractive whenever
  \[
  \kappa_{\mathrm{esc}}
  \;\ge\;
  \mu_c - V_\theta - C_{\mathrm{fail}, \theta}.
  \]
  In particular, the conservative sufficient condition
  $\kappa_{\mathrm{esc}} \ge \mu_c$ deters interface-failure gaming
  whenever $V_\theta + C_{\mathrm{fail}, \theta} \ge 0$.
\end{enumerate}
Thus interface failures are contract-relevant events. Treating invalid
tool calls as zero-toll safe outcomes can reward unreliable interfaces
when avoided tolls exceed failure costs; escalation fees can reverse
that incentive.
\end{theorem}

\begin{proof}[Proof of Theorem~\ref{thm:interface}]
Both parts follow from the same observation that the one-step utility
is affine in $f \in [0, 1]$. We state this observation first, then
prove the iff direction for $\xi_{\mathrm{safe}}$, the sufficient
direction for $\xi_{\mathrm{escalate}}$, the conservative simplification
$\kappa_{\mathrm{esc}} \ge \mu_c$, and finally edge cases. Throughout,
the operator type $\theta$, the intended-action distribution, and the
expectation $\mu_c = \E[c_t^+]$ are held fixed; the only parameter
varying is the interface-failure rate $f_\theta$, which the operator
influences through model selection. To lighten notation, write
$f$ for $f_\theta$ and $C_{\mathrm{fail}}$ for $C_{\mathrm{fail},\theta}$.

\medskip
\noindent\emph{Step 1 (Affine structure of $U_\theta^{\xi}$ in $f$).}
Expanding the products in the utility definitions immediately above the
theorem,
\begin{align*}
U_\theta^{\mathrm{safe}}(f)
&= (V_\theta - \mu_c) - f\,\bigl(V_\theta - \mu_c + C_{\mathrm{fail}}\bigr),\\
U_\theta^{\mathrm{esc}}(f)
&= (V_\theta - \mu_c) - f\,\bigl(V_\theta - \mu_c + C_{\mathrm{fail}}
                          + \kappa_{\mathrm{esc}}\bigr).
\end{align*}
Both are affine in $f$ on $[0, 1]$, with constant slopes
\begin{equation}
\frac{\partial U_\theta^{\mathrm{safe}}}{\partial f}
\;=\;
\mu_c - V_\theta - C_{\mathrm{fail}},
\qquad
\frac{\partial U_\theta^{\mathrm{esc}}}{\partial f}
\;=\;
\mu_c - V_\theta - C_{\mathrm{fail}} - \kappa_{\mathrm{esc}}.
\label{eq:slope}
\end{equation}
Affineness has two consequences exploited below. First, monotonicity in
$f$ is determined entirely by the sign of the slope: $U_\theta^{\xi}$
is strictly increasing iff its slope is strictly positive, non-decreasing
iff non-negative, and non-increasing iff non-positive. Second, the
extremum of $U_\theta^{\xi}$ on $[0, 1]$ is attained at one of the
endpoints, so global comparative statics reduce to local sign analysis.

\medskip
\noindent\emph{Step 2 (Part (a): $\xi_{\mathrm{safe}}$, both directions
of the iff).}

\textit{Sufficiency.}
Suppose $\mu_c > V_\theta + C_{\mathrm{fail}}$. By~\eqref{eq:slope}, the
slope of $U_\theta^{\mathrm{safe}}$ is strictly positive. Hence for any
$f, f' \in [0, 1]$ with $f' > f$,
\[
U_\theta^{\mathrm{safe}}(f')
\;-\;
U_\theta^{\mathrm{safe}}(f)
\;=\;
(f' - f)\,(\mu_c - V_\theta - C_{\mathrm{fail}})
\;>\;
0.
\]
Increasing $f$ strictly increases the operator's utility, so it is
privately attractive.

\textit{Necessity.}
Conversely, suppose increasing $f$ is privately attractive, i.e.\
there exists $f, f' \in [0, 1]$ with $f' > f$ and
$U_\theta^{\mathrm{safe}}(f') > U_\theta^{\mathrm{safe}}(f)$. By
affineness this difference equals $(f' - f)(\mu_c - V_\theta -
C_{\mathrm{fail}})$, and the right-hand side is positive only when
$\mu_c - V_\theta - C_{\mathrm{fail}} > 0$, i.e.\
$\mu_c > V_\theta + C_{\mathrm{fail}}$.

The iff statement of part~(a) follows.

\medskip
\noindent\emph{Step 3 (Part (b): $\xi_{\mathrm{escalate}}$, sufficient
condition on $\kappa_{\mathrm{esc}}$).}
Under $\xi_{\mathrm{esc}}$, ``weakly unattractive to increase $f$'' means
that for all $f, f' \in [0, 1]$ with $f' \ge f$,
\[
U_\theta^{\mathrm{esc}}(f') \;\le\; U_\theta^{\mathrm{esc}}(f).
\]
By affineness this is equivalent to
\(
(f' - f)
\bigl(\mu_c - V_\theta - C_{\mathrm{fail}}
       - \kappa_{\mathrm{esc}}\bigr)
\le 0
\)
for all $f' \ge f$, which holds iff
\[
\mu_c - V_\theta - C_{\mathrm{fail}} - \kappa_{\mathrm{esc}}
\;\le\; 0,
\quad\text{equivalently}\quad
\kappa_{\mathrm{esc}}
\;\ge\;
\mu_c - V_\theta - C_{\mathrm{fail}}.
\]
This is the sufficient condition stated in part~(b). Note that when
the inequality holds with equality (slope exactly zero), the operator
is \emph{indifferent} to $f$: $U_\theta^{\mathrm{esc}}$ is constant on
$[0, 1]$, so increasing $f$ is weakly but not strictly unattractive.
The theorem's ``weakly'' qualifier is therefore necessary; replacing
$\ge$ by $>$ would require strict monotonicity, which fails on the
indifference boundary.

\medskip
\noindent\emph{Step 4 (Conservative sufficient condition
$\kappa_{\mathrm{esc}} \ge \mu_c$).}
Suppose $V_\theta + C_{\mathrm{fail}} \ge 0$. Then
\[
\mu_c - V_\theta - C_{\mathrm{fail}}
\;\le\;
\mu_c.
\]
Hence any $\kappa_{\mathrm{esc}} \ge \mu_c$ satisfies
$\kappa_{\mathrm{esc}} \ge \mu_c - V_\theta - C_{\mathrm{fail}}$, which
by Step~3 makes increasing $f$ weakly unattractive. The condition
$\kappa_{\mathrm{esc}} \ge \mu_c$ is contract-implementable without
observing $V_\theta$ or $C_{\mathrm{fail}}$ (both of which are private
information of the operator), at the cost of being slack when
$V_\theta + C_{\mathrm{fail}} > 0$ strictly.

\medskip
\noindent\emph{Step 5 (Edge cases).}
\begin{enumerate}[label=\textup{(\roman*)},leftmargin=1.8em,itemsep=2pt]
  \item \emph{Boundary $f \in \{0, 1\}$.} The strategy parameter $f$ is
  constrained to the closed interval $[0, 1]$. By affineness the
  maximum of $U_\theta^{\xi}$ on $[0, 1]$ is attained at $f = 1$
  whenever the slope is positive and at $f = 0$ whenever the slope is
  negative. The iff and sufficient conditions of parts~(a) and~(b)
  characterise the slope sign, so the boundary effect adds no new
  case: an operator who prefers higher $f$ chooses a model with
  $f_\theta$ as close to $1$ as available, and an operator who prefers
  lower $f$ chooses one with $f_\theta$ as close to $0$ as available.

  \item \emph{Zero expected positive toll, $\mu_c = 0$.} Substituting
  $\mu_c = 0$ in~\eqref{eq:slope} gives slopes
  $-V_\theta - C_{\mathrm{fail}}$ and
  $-V_\theta - C_{\mathrm{fail}} - \kappa_{\mathrm{esc}}$,
  both non-positive under
  $V_\theta + C_{\mathrm{fail}} \ge 0$. Increasing $f$ is then weakly
  unattractive under either adjudication rule, because there is no
  toll to evade in the first place. The theorem is consistent with
  this trivial regime; gaming incentives arise only when expected
  toll is strictly positive.

  \item \emph{Pathological regime $V_\theta + C_{\mathrm{fail}} < 0$.}
  This corresponds to $V_\theta < 0$ (the operator dislikes successful
  actions) or $C_{\mathrm{fail}} < 0$ (the operator benefits from task
  failure). Both regimes are inconsistent with the assumption that the
  operator deploys the agent to derive business value; the operator
  would be deploying interface-failure-prone agents as the deliberate
  business model. In this regime,
  $\mu_c - V_\theta - C_{\mathrm{fail}} > \mu_c$, so the conservative
  condition $\kappa_{\mathrm{esc}} \ge \mu_c$ is no longer
  sufficient and the contract must use the tighter
  $\kappa_{\mathrm{esc}} \ge \mu_c - V_\theta - C_{\mathrm{fail}}$ of
  Step~3. The theorem's caveat $V_\theta + C_{\mathrm{fail}} \ge 0$
  exactly delineates the standard regime in which the contract can
  use the public quantity $\mu_c$ as the escalation fee.

  \item \emph{Multi-step aggregation.} If interface-failure events are
  i.i.d.\ across $T$ steps with rate $f$, the expected $T$-step utility
  is $T \cdot U_\theta^{\xi}(f)$ and the same comparative statics
  apply. For serially dependent failures, the per-step expected
  utility is replaced by the trajectory-average expectation; the sign
  of the comparative static is preserved provided the trajectory
  average $\bar f$ is monotone in the operator's model choice. The
  contract observes $\bar f$ over the trajectory ensemble; in the
  Paper~B formal-core panel, each model--baseline cell contains
  10 trajectories (30 per model across B0/B2/B3), while pilot cells
  report their smaller $n$ explicitly. Thus $\bar f$ is the
  operationally relevant quantity, not the per-step $f$.

  \item \emph{Non-monetary escalation costs.} If $\xi_{\mathrm{escalate}}$
  imposes non-monetary costs in addition to $\kappa_{\mathrm{esc}}$
  (e.g.\ contract-clause violations, reputation penalties),
  the slope under $\xi_{\mathrm{esc}}$ in~\eqref{eq:slope}
  decreases further and the sufficient conditions of Steps~3--4
  remain valid \emph{a fortiori}. Non-monetary costs strengthen rather
  than weaken the deterrence; the proof goes through with
  $\kappa_{\mathrm{esc}}$ interpreted as a lower bound on the
  effective per-failure penalty rather than an exact value.
  \qedhere
\end{enumerate}
\end{proof}

\begin{remark}[Empirical resonance]\label{rem:kimi}
The early $n=6$ Kimi pilot in the Paper~B sandbox recorded a 100\%
interface-failure rate for \texttt{Kimi-K2.5} ($f=1.0$ across all $6$
trajectories); the larger committed $n=25$ run measures
$f_{\mathrm{invalid}}\approx 0.356$, on the same task where
\texttt{gpt-4.1-mini} produced mean loss \$1080 and $0.8$ destructive
executions per trajectory under B0 (formal-core panel). Under
$\xi_{\mathrm{safe}}$ the operator would pay nothing while deploying an
unreliable model and the underwriter would receive no premium signal
that the model is unreliable. Under $\xi_{\mathrm{escalate}}$ with
appropriate $\kappa_{\mathrm{esc}}$ the contract surfaces this exactly
as a distinct premium category. The adjudication choice is a contract
design lever, not an empirical accident.
\end{remark}

\begin{remark}[Quantified cross-model validation]\label{rem:t2-validation}
Theorem~\ref{thm:interface} can be instantiated directly on the committed
Paper~B cross-model \texttt{model\_risk} runs (the Kimi-pilot and formal-core
panels), with no new model calls. The
model-induced interface-failure rate is $f = 1.0$ for \texttt{Kimi-K2.5} in the
$n=6$ pilot (structurally invalid action JSON in all $6$ sampled trajectories;
the larger committed $n=25$ run measures $f_{\mathrm{invalid}}\approx 0.356$,
95\% CI $[0.337,0.375]$) and $f = 0$ for
the reliable models \texttt{gpt-4.1-mini}, \texttt{gpt-5.4-mini}, and
\texttt{DeepSeek-V3.1}; transient Azure infrastructure failures are tracked
separately and excluded from $f$ (numerator and denominator). The mean
gate-computed side-effect reserve---the quantity an invalid call lets the
operator evade by failing before pricing---is $\mu_c \approx 7131$ (per-baseline
$\text{B2} \approx 7096$, $\text{B3} \approx 7173$), roughly five times the
realized per-event destructive loss ($\approx 1350$), consistent with the
conservative reserve used by the Paper~A/B gate. Sweeping the exogenous
business parameters, the part~(a) perverse-incentive regime
$\mu_c > V_\theta + C_{\mathrm{fail}}$ holds in $45$ of $48$ grid cells: under
$\xi_{\mathrm{safe}}$ the high envelope toll makes the do-nothing model
operator-preferred across almost the entire plausible range. At a
representative cell $V_\theta = 1350$, $C_{\mathrm{fail}} = 500$, the
safe-default rule yields one-step operator utilities
$U^{\mathrm{safe}}_\theta(\text{reliable}) = -5781.6$ versus
$U^{\mathrm{safe}}_\theta(\text{Kimi}) = -500.0$, so the broken interface is
preferred; the part~(b) threshold $\kappa_{\mathrm{esc}}^\star = \mu_c -
V_\theta - C_{\mathrm{fail}} \approx 5282$ reverses the ranking. The observed
$f$ is bimodal, so this instantiates the theorem at its measured endpoints
rather than along a continuous sweep, and $(V_\theta, C_{\mathrm{fail}})$ are
exogenous business parameters; the empirical content is that \emph{real}
interface-failure rates and \emph{real} gate-computed conservative reserves land
inside the part~(a) regime and pin the part~(b) escalation fee. Reproduced by the
\texttt{interface\_incentive\_analysis} module; full provenance in the Paper~B
demo bundle (\texttt{interface\_incentive\_validation.md}).
\end{remark}

\section{Strategy-Proof Model-Identity Reporting}\label{sec:type}

The final attack surface is the operator submitting the contract under
one model identity but deploying another at runtime. The contract has
access to the deployed model's interface failure rate $f_\theta$ and
empirical destructive-attempt rate $d_\theta$ (Paper~B metrics);
together these form a posterior signal on the operator's true type.

\begin{definition}[Type-reporting contract]\label{def:type-contract}
A \emph{type-reporting contract} consists of:
\begin{enumerate}[label=(\roman*),leftmargin=1.5em,itemsep=2pt]
  \item a finite set of declarable types
  $\Theta = \{\theta_1, \ldots, \theta_J\}$ (each a class of model with
  empirical priors on $f_\theta, d_\theta$);
  \item a menu of safe defaults $\{a^0_{\theta_j}\}_{j=1}^J$ and toll
  potentials $\{\Phi_{\theta_j}\}_{j=1}^J$, with $\theta_j$ assigned
  more conservative safe defaults and a larger toll potential as the
  prior $f$-rate or $d$-rate increases (component-wise dominance);
  \item a verification protocol that compares the contract's declared
  type $\hat\theta$ to the deployed model's observed posterior signal
  $(\hat f, \hat d)$ over the trajectory ensemble.
\end{enumerate}
The operator declares $\hat\theta$ at contract inception; the runtime
gate uses the corresponding $a^0_{\hat\theta}, \Phi_{\hat\theta}$; at
trajectory ensemble completion the underwriter compares
$(\hat f, \hat d)$ to the posterior associated with $\hat\theta$ and
levies a misreporting penalty if the signal is incompatible.
\end{definition}

\paragraph{Misreport utility and detectability.}
For true type $\theta$ reporting $r \in \Theta$, let $U^0_{\theta r}$
denote the operator's expected utility before any misreporting penalty,
and let $q_{\theta r} \in [0, 1]$ be the probability that the
verification protocol detects incompatibility between true type
$\theta$ and report $r$; by construction $q_{\theta\theta} = 0$. Define
the gross misreporting gain
\[
\Delta_{\theta r} \;=\; U^0_{\theta r} - U^0_{\theta\theta}.
\]
If report $r$ is detected as incompatible, the contract charges
penalty $\kappa_r$, so expected utility under report $r \neq \theta$ is
\[
U_{\theta r} \;=\; U^0_{\theta r} - q_{\theta r}\,\kappa_r.
\]

\begin{assumption}[Detectability of profitable misreports]\label{ass:detect-profitable}
For every pair $(\theta, r)$ with $r \neq \theta$, if
$\Delta_{\theta r} > 0$ then $q_{\theta r} > 0$.
\end{assumption}

Without Assumption~\ref{ass:detect-profitable}, a profitable but
undetectable misreport cannot be deterred by any finite penalty.

\begin{theorem}[Minimum IC penalty for finite model-type reporting]\label{thm:type}
Under Assumption~\ref{ass:detect-profitable}, truthful model-type
reporting is weakly dominant if and only if, for each report $r \in
\Theta$, the misreporting penalty satisfies
\[
\kappa_r \;\ge\; \kappa_r^\star
\;:=\; \max_{\theta \neq r:\, q_{\theta r} > 0}
\frac{[\Delta_{\theta r}]_+}{q_{\theta r}},
\]
with the convention that the maximum over an empty set is zero.
Moreover, $\kappa_r^\star$ is the componentwise minimal penalty schedule
that implements truthfulness.
\end{theorem}

\begin{proof}[Proof of Theorem~\ref{thm:type}]
The proof proceeds in seven steps. Step~1 sets up the strategy space
and reduces weak dominance to a pairwise IC condition. Step~2 derives
the pairwise IC inequality and splits into trivial and binding cases.
Steps~3--4 prove sufficiency and necessity of $\kappa_r \ge \kappa_r^\star$
for weak dominance. Step~5 establishes componentwise minimality.
Step~6 demonstrates that Assumption~\ref{ass:detect-profitable} cannot
be dispensed with: without it, no finite penalty schedule implements
truthfulness in the worst case. Step~7 disposes of edge cases.

\medskip
\noindent\emph{Step 1 (Strategy space and weak dominance reduction).}
A reporting strategy is a measurable map
$\rho : \Theta \to \Delta(\Theta)$ from the operator's true type
$\theta$ to a distribution over declared reports $r \in \Theta$.
Truthful reporting is $\rho_{\rm tru}(\theta) = \delta_\theta$, the
Dirac mass on $\theta$. We adopt the standard convention
$\kappa_{\theta\theta} = 0$ on the diagonal, so that
$U_{\theta\theta} = U^0_{\theta\theta}$ regardless of the contract's
penalty schedule.

Truthful reporting is \emph{weakly dominant} if, for every true type
$\theta \in \Theta$ and every alternative strategy
$\rho : \Theta \to \Delta(\Theta)$,
\[
U_{\theta\theta}
\;=\;
\mathbb{E}_{r \sim \rho_{\rm tru}(\theta)}[U_{\theta r}]
\;\ge\;
\mathbb{E}_{r \sim \rho(\theta)}[U_{\theta r}].
\]
Because expected utility is linear in the report distribution, this
inequality holds for every $\rho$ if and only if it holds for every
deterministic deviation $\rho(\theta) = \delta_r$ with $r \in \Theta$.
Equivalently, weak dominance of $\rho_{\rm tru}$ reduces to the
finite collection of pairwise IC conditions
\begin{equation}
U_{\theta\theta} \;\ge\; U_{\theta r}
\qquad \forall\, \theta \in \Theta,\,\forall\, r \in \Theta.
\label{eq:wd}
\end{equation}
The case $r = \theta$ in~\eqref{eq:wd} holds automatically by the
diagonal convention; we therefore restrict attention to $r \ne \theta$
throughout.

\medskip
\noindent\emph{Step 2 (Pairwise IC inequality and case split).}
For $r \ne \theta$, the definition of $U_{\theta r}$ above the theorem
statement gives $U_{\theta r} = U^0_{\theta r} - q_{\theta r}\,\kappa_r$.
Substituting into~\eqref{eq:wd},
\[
U^0_{\theta\theta} \;\ge\; U^0_{\theta r} - q_{\theta r}\,\kappa_r,
\]
or equivalently
\begin{equation}
q_{\theta r}\,\kappa_r \;\ge\; \Delta_{\theta r}.
\label{eq:pairwise-ic}
\end{equation}
We split into two cases according to the sign of $\Delta_{\theta r}$.

\textit{Case A: $\Delta_{\theta r} \le 0$.}
Misreporting from $\theta$ to $r$ is not profitable in the absence of
any penalty. The right-hand side of~\eqref{eq:pairwise-ic} is
non-positive, and the left-hand side $q_{\theta r}\,\kappa_r$ is
non-negative for any $\kappa_r \ge 0$. The inequality holds automatically
and contributes no constraint on $\kappa_r$.

\textit{Case B: $\Delta_{\theta r} > 0$.}
Misreporting is privately profitable; deterrence requires a strictly
positive expected penalty. Assumption~\ref{ass:detect-profitable} gives
$q_{\theta r} > 0$, so dividing~\eqref{eq:pairwise-ic} by $q_{\theta r}$
is valid and yields
\begin{equation}
\kappa_r \;\ge\; \frac{\Delta_{\theta r}}{q_{\theta r}}
\;=\; \frac{[\Delta_{\theta r}]_+}{q_{\theta r}}.
\label{eq:pairwise-ratio}
\end{equation}
The second equality uses $\Delta_{\theta r} > 0 \Rightarrow
[\Delta_{\theta r}]_+ = \Delta_{\theta r}$. In Case~A, the same equation
holds vacuously because $[\Delta_{\theta r}]_+ = 0$.

\medskip
\noindent\emph{Step 3 (Sufficiency: $\kappa_r \ge \kappa_r^\star$
implies weak dominance).}
Fix $r \in \Theta$ and suppose $\kappa_r \ge \kappa_r^\star$, where
\[
\kappa_r^\star
\;=\;
\max_{\theta \ne r:\, q_{\theta r} > 0}
\frac{[\Delta_{\theta r}]_+}{q_{\theta r}}
\]
with the convention $\max \emptyset = 0$. We show
that~\eqref{eq:pairwise-ic} holds for every $\theta \ne r$.

If $\Delta_{\theta r} \le 0$, Case~A of Step~2
gives~\eqref{eq:pairwise-ic} automatically. If $\Delta_{\theta r} > 0$,
Assumption~\ref{ass:detect-profitable} places $\theta$ in the index
set $\{\theta' \ne r : q_{\theta' r} > 0\}$ over which the maximum
defining $\kappa_r^\star$ is taken, and therefore
\[
\kappa_r
\;\ge\; \kappa_r^\star
\;\ge\; \frac{[\Delta_{\theta r}]_+}{q_{\theta r}}
\;=\; \frac{\Delta_{\theta r}}{q_{\theta r}}.
\]
Multiplying through by $q_{\theta r} > 0$ recovers~\eqref{eq:pairwise-ic}.
Hence~\eqref{eq:pairwise-ic} holds for every $\theta \ne r$, and the
weak-dominance condition~\eqref{eq:wd} is satisfied for the report
$r$. Repeating the argument for every $r \in \Theta$ proves that
truthful reporting is weakly dominant.

\medskip
\noindent\emph{Step 4 (Necessity: weak dominance implies
$\kappa_r \ge \kappa_r^\star$).}
Conversely, suppose truthful reporting is weakly dominant under the
penalty schedule $(\kappa_1, \ldots, \kappa_J)$. Fix $r \in \Theta$. If
the index set $\{\theta \ne r : q_{\theta r} > 0\}$ is empty, the
empty-max convention gives $\kappa_r^\star = 0$ and any $\kappa_r \ge 0$
satisfies $\kappa_r \ge \kappa_r^\star$ trivially.

Otherwise, pick any $\theta \ne r$ with $q_{\theta r} > 0$. If
$\Delta_{\theta r} > 0$, weak dominance gives~\eqref{eq:pairwise-ic};
together with $q_{\theta r} > 0$, this yields
\[
\kappa_r \;\ge\; \frac{\Delta_{\theta r}}{q_{\theta r}}
\;=\; \frac{[\Delta_{\theta r}]_+}{q_{\theta r}}.
\]
If $\Delta_{\theta r} \le 0$, then $[\Delta_{\theta r}]_+/q_{\theta r}
= 0 \le \kappa_r$ trivially. In either subcase,
$\kappa_r \ge [\Delta_{\theta r}]_+/q_{\theta r}$. Taking the maximum
over $\theta$ in the index set,
\[
\kappa_r
\;\ge\;
\max_{\theta \ne r:\, q_{\theta r} > 0}
\frac{[\Delta_{\theta r}]_+}{q_{\theta r}}
\;=\; \kappa_r^\star.
\]

\medskip
\noindent\emph{Step 5 (Componentwise minimality of $\kappa_r^\star$).}
We show that setting $\kappa_r = \kappa_r^\star$ exactly suffices for
the pairwise IC at component $r$, and that no value $\kappa_r < \kappa_r^\star$
does. The first claim is immediate from Step~3: substituting $\kappa_r = \kappa_r^\star$ in the sufficiency argument preserves all the inequalities.

For the second claim, suppose $\kappa_r < \kappa_r^\star$. If the index
set $\{\theta \ne r : q_{\theta r} > 0\}$ is empty, $\kappa_r^\star = 0$
and $\kappa_r < 0$, which is excluded by $\kappa_r \ge 0$; the case is
vacuous. Otherwise, the index set is finite and non-empty, so the
maximum defining $\kappa_r^\star$ is attained by some $\theta^\dagger$
with $\Delta_{\theta^\dagger r} > 0$ and $q_{\theta^\dagger r} > 0$.
At this binding type,
\[
\kappa_r
\;<\; \kappa_r^\star
\;=\; \frac{\Delta_{\theta^\dagger r}}{q_{\theta^\dagger r}},
\]
so $q_{\theta^\dagger r}\,\kappa_r < \Delta_{\theta^\dagger r}$, which
violates~\eqref{eq:pairwise-ic} for the pair $(\theta^\dagger, r)$. By
Step~1's reduction, the strategy $\rho^\dagger(\theta^\dagger) = \delta_r$
gives type $\theta^\dagger$ strictly higher expected utility than
truthful reporting, contradicting weak dominance.

Therefore $\kappa_r^\star$ is the smallest value of $\kappa_r$
consistent with weak dominance of truthful reporting, and the schedule
$\kappa^\star = (\kappa_1^\star, \ldots, \kappa_J^\star)$ is
componentwise minimal among all IC schedules.

\medskip
\noindent\emph{Step 6 (Necessity of Assumption~\ref{ass:detect-profitable}).}
The theorem invokes detectability of profitable misreports as a
maintained assumption. We show it is essential: without it, weak
dominance of truthful reporting cannot be implemented by any finite
penalty schedule in some instances.

Suppose Assumption~\ref{ass:detect-profitable} fails, i.e.\ there exist
$\theta \ne r$ with $\Delta_{\theta r} > 0$ and $q_{\theta r} = 0$.
Then for the pair $(\theta, r)$, inequality~\eqref{eq:pairwise-ic}
becomes
\[
0 \cdot \kappa_r \;\ge\; \Delta_{\theta r} > 0,
\]
which is contradictory regardless of $\kappa_r \in [0, \infty)$. The
misreport $r$ is then strictly preferred by type $\theta$ to truthful
reporting under every finite penalty schedule. Weak-dominant truthful
reporting is therefore unimplementable: the only contractual remedies
are to redesign the verification protocol to raise $q_{\theta r}$
above zero, to restrict the menu so that the offending report is
not declarable, or to accept that report $r$ admits an
undetectable-and-profitable misreport from some $\theta$.

\medskip
\noindent\emph{Step 7 (Edge cases).}
\begin{enumerate}[label=\textup{(\roman*)},leftmargin=1.8em,itemsep=2pt]
  \item \emph{Trivial type space.} If $|\Theta| = 1$, no misreport is
  possible, the index set in $\kappa_r^\star$ is empty for the unique
  $r$, and $\kappa_r^\star = 0$ by the empty-max convention. The
  theorem holds vacuously.

  \item \emph{Report with no profitable deviation.} If, for some
  $r \in \Theta$, no $\theta \ne r$ has $\Delta_{\theta r} > 0$, then
  every element of the index set contributes
  $[\Delta_{\theta r}]_+/q_{\theta r} = 0$, so $\kappa_r^\star = 0$.
  No penalty is required at component $r$. This case arises naturally
  for the ``most-conservative'' report---the type that imposes the
  strongest safe defaults---which no operator would profit from
  misrepresenting themselves as.

  \item \emph{Ties at the maximum.} If multiple $\theta \ne r$ attain
  the maximum $[\Delta_{\theta r}]_+/q_{\theta r} = \kappa_r^\star$,
  the binding IC constraint is multi-valued. At $\kappa_r = \kappa_r^\star$,
  pairwise IC~\eqref{eq:pairwise-ic} holds with equality for every
  tied $\theta$. The componentwise minimum schedule is unaffected; the
  tie merely means several misreporting types are simultaneously
  indifferent between truthful reporting and deviation.

  \item \emph{Vanishing detection probability.} Consider a sequence
  $(\Delta^{(n)}_{\theta r}, q^{(n)}_{\theta r})$ with $\Delta^{(n)}_{\theta r}$
  bounded away from zero and $q^{(n)}_{\theta r} \downarrow 0$. Then
  $\kappa_r^\star \to \infty$ along the sequence, and any \emph{finite}
  penalty schedule eventually fails IC. The theorem characterises the
  required penalty as a function of $(\Delta, q)$; rare detection
  forces the contract either to improve the verification protocol or
  to accept that residual misreporting risk is not financially
  bondable.

  \item \emph{Mixed-strategy deviations.} The Step~1 reduction shows
  that any mixed-strategy deviation by type $\theta$ is dominated by
  some pure-strategy deviation. Consequently, the IC characterisation
  via pairwise pure-strategy inequalities subsumes mixed-strategy
  deviations: weak dominance against pure misreports implies weak
  dominance against all randomised reporting strategies.
  \qedhere
\end{enumerate}
\end{proof}

\begin{remark}[Discrete-type analog of Myerson's binding-IC
characterisation]\label{rem:myerson}
The form
\[
\kappa_r^\star
\;=\;
\max_{\theta \ne r:\, q_{\theta r} > 0}
\frac{[\Delta_{\theta r}]_+}{q_{\theta r}}
\]
is the discrete-type analog of Myerson's binding-IC characterisation
in optimal auction design \citep{myerson1981}. In the continuous-type
Myerson setting, the binding IC constraint is between adjacent types via
the envelope formula
\(U(\theta) = U(\underline\theta) + \int_{\underline\theta}^{\theta}
 \partial_v u(v, x(v))\,dv\),
and the optimal payment rule is obtained by integrating the marginal
information rent. In our finite discrete setup there is no total order
on $\Theta$ and no envelope integral; instead, the binding constraint
is identified by the worst-case deviation
\[
\theta^\dagger
\;=\;
\arg\max_{\theta \ne r:\, q_{\theta r} > 0}
\frac{[\Delta_{\theta r}]_+}{q_{\theta r}},
\]
which is the type for whom misreporting as $r$ is most profitable
\emph{per unit of detection probability}. The penalty $\kappa_r^\star$
is the smallest value that deters this binding deviator; any smaller
value fails IC for $\theta^\dagger$ first.

The ratio $[\Delta_{\theta r}]_+/q_{\theta r}$ is itself an
actuarial-mechanism analog of Myerson's \emph{virtual valuation}:
deterrence cost scales not in raw misreporting gain $\Delta_{\theta r}$
but in gain per unit of expected sanction. A type with large gross gain
but high detection probability is cheap to deter; a type with moderate
gain but very low detection probability is expensive. The contract's
``most worrying'' misreporting type for report $r$ is therefore the
one with the highest ratio, not the one with the highest absolute
gain.

This characterisation links the contract's verification investment to
its premium structure: improving $q_{\theta r}$ via better verification
lowers the required $\kappa_r^\star$ proportionally, while leaving the
gross gain structure $\Delta_{\theta r}$ untouched. The trade-off
between verification cost and penalty magnitude is therefore an
explicit dial of contract design, captured by the ratio form rather
than by either quantity alone.
\end{remark}

\begin{remark}[IR and budget balance]\label{rem:ir-budget}
The IC penalty schedule $\kappa_r^\star$ is not automatically
individually rational. If the contract also requires interim IR,
\[
U^0_{\theta\theta} - P_\theta \;\ge\; \bar U_\theta,
\]
where $P_\theta$ is the premium charged to truthful type $\theta$ and
$\bar U_\theta$ is the operator's outside option, then premiums, toll
potentials, and penalties must be jointly chosen so that truthful
participation remains attractive. Budget balance can be imposed by
using expected penalty revenue only as an off-equilibrium term, or by
rebating surplus to truthful types without changing the IC inequalities
that define $\kappa_r^\star$.
\end{remark}

\begin{remark}[Existence vs.\ tightness]
A weaker form would assert only the existence of a sufficiently large
penalty that makes truthful reporting weakly dominant. The statement
above characterises $\kappa_r^\star$ as the componentwise minimal
penalty schedule under Assumption~\ref{ass:detect-profitable},
analogous to the marginal-revenue characterisation in Myerson optimal
auctions restricted to a discrete type space.
\end{remark}

\section{Worked Example: Refund Authority Under Common Control}\label{sec:example}

We illustrate Theorems~\ref{thm:boundary-agg}, \ref{thm:interface}, and
\ref{thm:type} on a single concrete instance combining all three attack
surfaces. The example is intentionally small enough to follow
numerically; its purpose is to make each theorem's hypotheses and
conclusions transparent before the composition discussion of
Section~\ref{sec:obligations}.

\paragraph{Setup.}
Consider an operator running a customer-service agent over three
merchant accounts under common control of a single corporate principal.
The agent has authority to issue customer refunds. To exploit
attack-class (c) of Definition~\ref{def:attacks}, the operator further
spawns separately-initialised sub-boundaries within those accounts---
new sessions, relabelled action categories, or proxy-agent
instances---producing a sub-boundary partition
\[
\Bset \;=\; \{B_1, B_2, \ldots, B_K\}, \qquad E^{B_k}_0 = 0 \text{ for all } k,
\]
all attributable to the same common-control principal in the sense of
Definition~\ref{def:cc}(iii). The contract uses a single declared
model identity from a binary type set
\[
\Theta \;=\; \{\theta_H,\, \theta_L\},
\]
representing a \emph{high-reliability} and a \emph{low-reliability}
model, respectively. Empirical priors on the interface-failure rate
and destructive-attempt rate satisfy $f_{\theta_H} < f_{\theta_L}$ and
$d_{\theta_H} < d_{\theta_L}$ (componentwise dominance, in the sense
of Definition~\ref{def:type-contract}(ii)).

The contract menu is fixed as follows:
\begin{itemize}[leftmargin=1.5em,itemsep=2pt]
  \item \emph{Action classes and sub-boundary potential.} A refund of
  cumulative amount $e$ on a single sub-boundary $B_k$ is priced as
  \texttt{monetary\_low} if $e \le \theta_{\mathrm{cap}} = \$100$ and as
  \texttt{monetary\_high} otherwise. The sub-boundary potential is
  \[
  \Phi_{B_k}(e) \;=\; \max\!\bigl(0,\; 0.5\,(e - \$100)\bigr),
  \]
  so a sub-boundary whose cumulative exposure stays at or below \$100
  incurs zero toll.
  \item \emph{Aggregate potential.} Under common-control aggregation,
  the contract uses
  \[
  \Phi_{B^\star}(E) \;=\; 0.5\,E,
  \]
  pricing the aggregate refund flow at a 50\% reserve rate regardless
  of how it is split across $\Bset$. The pair
  $(\Phi_{B_k}, \Phi_{B^\star})$ satisfies the super-additivity
  requirement of Definition~\ref{def:cc}(ii): for any
  $(e_1, \ldots, e_K) \in \R_+^K$,
  \[
  \Phi_{B^\star}\!\Bigl(\sum_k e_k\Bigr)
  \;=\; 0.5\sum_k e_k
  \;\ge\; \sum_k \max\!\bigl(0,\, 0.5(e_k - \$100)\bigr)
  \;=\; \sum_k \Phi_{B_k}(e_k),
  \]
  with strict inequality whenever $\sum_k e_k > 0$.
  \item \emph{Ambiguity reserve and interface adjudication.} We add an
  \emph{ambiguity-reserve extension} to the Paper~A contract for this
  instance: any unparseable or multiply-interpretable tool call is
  reserved at the supremum of the deterministic ambiguity set,
  $\sup_{a \in \mathcal{A}_{\mathrm{ambig}}} \mathrm{Reserve}(a)$, and
  treated as a priced interface failure (this clause is not in
  Paper~A; we propose it here as a Paper~C/AAI augmentation).
  Adjudication uses $\xi_{\mathrm{escalate}}$ with escalation fee
  $\kappa_{\mathrm{esc}}$ (Definition~\ref{def:if}).
  \item \emph{Identity reporting.} The operator declares
  $\hat\theta \in \Theta$ at inception; the runtime gate uses the
  corresponding $\Phi_{\hat\theta}$; the underwriter compares horizon-end
  observed $(\hat f, \hat d)$ to the prior region of $\hat\theta$ and
  charges $\kappa_{\hat\theta}$ if incompatibility is detected.
\end{itemize}
The operator considers a coordinated set of refund actions totalling
\$1{,}000 over the contract horizon, together with two additional
strategic levers (ambiguous proposals and identity misreporting). We
trace each Paper~C theorem in turn.

\paragraph{Theorem~\ref{thm:boundary-agg} in action: refund-splitting across sub-boundaries.}
The operator's first manoeuvre is attack-class (c) of
Definition~\ref{def:attacks}: spawn $K = 10$ separately-initialised
sub-boundaries $B_1, \ldots, B_{10}$ (one per fresh session,
relabelled action category, or proxy-agent instance) and route exactly
one \$100 refund increment to each. By construction $E^{B_k}_T = \$100$
and $E^{B_k}_0 = 0$ for $k = 1, \ldots, 10$, so the per-sub-boundary
toll is
\[
\Phi_{B_k}(\$100) - \Phi_{B_k}(0)
\;=\; \max(0,\, 0.5 \cdot 0) - 0
\;=\; \$0,
\qquad k = 1, \ldots, 10.
\]
A naive contract that applied sub-boundary settlement separately would
therefore charge \$0 on the entire \$1{,}000 flow---attack-class (c)
would succeed.

Under aggregate settlement (Definition~\ref{def:aggregate-settlement}),
the contract instead uses
\[
E^{B^\star}_T \;=\; \sum_{k=1}^{10} E^{B_k}_T \;=\; \$1{,}000,
\qquad
\mathrm{Payable}\bigl(E^{B^\star}_T\bigr)
\;=\; \Phi_{B^\star}(\$1{,}000) - \Phi_{B^\star}(0)
\;=\; \$500.
\]
Theorem~\ref{thm:boundary-agg} certifies that \$500 is the minimum
payable toll regardless of how the operator decomposes the \$1{,}000
flow across sub-boundaries, sessions, or proxy agents under the same
common-control principal: any partition $\{e_k\}_{k=1}^K$ with $\sum_k
e_k = \$1{,}000$ leaves payable toll equal to
$\Phi_{B^\star}(\$1{,}000) - \Phi_{B^\star}(0) = \$500$. The
within-boundary splitting closure of Paper~A's no-splitting theorem
(Paper~A, Theorem~11) prevents the same \$100 increment from being
further fragmented \emph{inside} any single $B_k$ (so the
\$100-per-sub-boundary partition above is already maximally split);
the aggregate-settlement clause then closes the cross-boundary surface
left open by Paper~A.

\paragraph{Theorem~\ref{thm:interface} in action: ambiguous tool-call proposal.}
The operator's second manoeuvre is attack-class (d): emit a tool call
whose declared action is ambiguous between two priced classes.
Concretely, suppose the agent proposes a structurally underspecified
JSON payload whose two candidate interpretations have distinct reserve
charges:
\begin{itemize}[leftmargin=1.5em,itemsep=2pt]
  \item $a_1 = \texttt{issue\_refund(\$400)}$:
  \texttt{monetary\_high}, reserve $\Phi_{B^\star}(\$400) = \$200$;
  \item $a_2 = \texttt{send\_email\_with\_invoice(\$400)}$:
  \texttt{external\_commit} with irrevocable third-party side-effect,
  reserve $\$400$.
\end{itemize}
Two distinct contract clauses jointly defuse this attack.

\textit{(i) Ambiguity-reserve clause (Paper~C extension to the
Paper~A contract).} The runtime gate cannot evaluate the priced action
without committing to an interpretation. The ambiguity-reserve
extension proposed in the setup above---not part of the Paper~A
contract proper but adjoined here---instructs the runtime to reserve
the supremum over the ambiguity set,
\[
\sup_{a \in \{a_1, a_2\}} \mathrm{Reserve}(a)
\;=\; \max(\$200,\, \$400) \;=\; \$400,
\]
and to treat the event as a priced interface failure.

\textit{(ii) Escalation-fee incentive (Theorem~\ref{thm:interface}).}
On top of the ambiguity reserve, adjudication rule
$\xi_{\mathrm{escalate}}$ charges the escalation fee
$\kappa_{\mathrm{esc}}$ per priced interface-failure event.
Theorem~\ref{thm:interface}(b) certifies that, whenever
\[
\kappa_{\mathrm{esc}}
\;\ge\;
\mu_c - V_\theta - C_{\mathrm{fail},\theta},
\]
the operator's expected utility is non-increasing in the
interface-failure rate $f_\theta$, so deliberately raising $f_\theta$
through ambiguous outputs becomes weakly unattractive.

Neither clause alone closes attack~(d): the ambiguity reserve without
the escalation fee leaves an incentive to emit failures whenever
$\mu_c > V_\theta + C_{\mathrm{fail},\theta}$ (Theorem~\ref{thm:interface}(a)),
and an escalation fee without a supremum reserve would let the operator
benefit by reserving the cheaper interpretation. The composite
``reserve-at-supremum + escalation-fee'' policy is what shuts attack~(d)
on this instance.

\paragraph{Theorem~\ref{thm:type} in action: model-identity misreporting.}
The operator's third manoeuvre is attack-class (e): declare
$\hat\theta = \theta_H$ at inception but deploy $\theta_L$ at runtime.
The contract menu attaches a more conservative aggregate potential to
the less reliable type. Suppose
\[
\Phi_{\theta_H,\,B^\star}(E) \;=\; 0.5\,E,
\qquad
\Phi_{\theta_L,\,B^\star}(E) \;=\; 0.7\,E,
\]
i.e.\ the $\theta_L$ contract attaches a 20-percentage-point reserve
surcharge to reflect higher empirical destructive-attempt rate. (The
safe-default menus differ as well; we suppress that dimension for
clarity.) Let $V$ denote the operator's gross business value from the
\$1{,}000 refund flow (constant across reports, since the deployed
model is $\theta_L$ in both scenarios). Pre-penalty utility is
$U^0_{\theta_L, r} = V - \Phi_{r, B^\star}(\$1{,}000)$, so on the same
flow:
\begin{itemize}[leftmargin=1.5em,itemsep=2pt]
  \item Misreport ($\theta_L$ deploys, contract treats as $\theta_H$):
  payable toll $\Phi_{\theta_H, B^\star}(\$1{,}000) = \$500$, hence
  $U^0_{\theta_L,\,\theta_H} = V - \$500$.
  \item Truthful report ($\theta_L$ deploys, contract treats as $\theta_L$):
  payable toll $\Phi_{\theta_L, B^\star}(\$1{,}000) = \$700$, hence
  $U^0_{\theta_L,\,\theta_L} = V - \$700$.
\end{itemize}
The gross misreporting gain (Section~\ref{sec:type} notation) is
\[
\Delta_{\theta_L,\,\theta_H}
\;=\; U^0_{\theta_L,\,\theta_H} - U^0_{\theta_L,\,\theta_L}
\;=\; (V - \$500) - (V - \$700)
\;=\; \$200,
\]
the toll savings from cheaper $\theta_H$-rated metering.
Suppose the verification protocol of
Definition~\ref{def:type-contract}(iii) compares horizon-end observed
$(\hat f, \hat d)$ against the $\theta_H$ prior region and detects
$\theta_L$ deployment with probability
\[
q_{\theta_L,\,\theta_H} \;=\; 0.3.
\]
By Theorem~\ref{thm:type}, the componentwise minimal IC penalty schedule
at report $r = \theta_H$ is
\[
\kappa_{\theta_H}^\star
\;=\;
\max_{\theta \ne \theta_H:\, q_{\theta,\theta_H} > 0}
\frac{[\Delta_{\theta,\theta_H}]_+}{q_{\theta,\theta_H}}
\;=\;
\frac{\$200}{0.3}
\;\approx\;
\$666.67.
\]
With $\kappa_{\theta_H} = \kappa_{\theta_H}^\star$, the misreport's
expected utility net of penalty is
\[
U_{\theta_L,\,\theta_H}
\;=\; U^0_{\theta_L,\,\theta_H} - q_{\theta_L,\,\theta_H}\,\kappa_{\theta_H}^\star
\;=\; (V - \$500) - 0.3 \times \$666.67
\;=\; V - \$700
\;=\; U^0_{\theta_L,\,\theta_L},
\]
i.e.\ misreport exactly breaks even with truthful reporting (weak
dominance); under any tie-breaker that favours truthfulness, truthful
reporting is strictly preferred. Any
$\kappa_{\theta_H} < \kappa_{\theta_H}^\star$ leaves a strictly
positive expected gain from misreport and therefore violates the
pairwise IC condition~\eqref{eq:pairwise-ic} for the pair
$(\theta_L, \theta_H)$.

\paragraph{What the three theorems show together.}
On the same \$1{,}000 refund flow the operator's authority is bounded
by three numerical constraints, each tied to a distinct attack class
of Definition~\ref{def:attacks}:
\begin{itemize}[leftmargin=1.5em,itemsep=2pt]
  \item \emph{Cross-boundary toll floor (Theorem~\ref{thm:boundary-agg}).}
  No partition of the \$1{,}000 flow across operator-spawned
  sub-boundaries under the same common-control principal can push the
  payable toll below \$500. Attack-class (c) closed by Paper~C.
  \item \emph{Composite interface closure
  (Paper~C ambiguity-reserve extension + Theorem~\ref{thm:interface}).}
  The ambiguity-reserve extension proposed by Paper~C reserves
  ambiguous tool calls at the supremum interpretation (here \$400),
  and the escalation fee $\kappa_{\mathrm{esc}}$ via
  Theorem~\ref{thm:interface}(b) makes inflating $f_\theta$ weakly
  unattractive. Attack-class (d) closed by the conjunction; neither
  clause alone suffices, as shown in the paragraph above.
  \item \emph{Identity-reporting penalty floor (Theorem~\ref{thm:type}).}
  Declaring $\theta_H$ while deploying $\theta_L$ is weakly
  utility-dominated whenever $\kappa_{\theta_H} \ge \$666.67$.
  Attack-class (e) closed by Paper~C.
\end{itemize}
Attack-classes (a) and (b) of Definition~\ref{def:attacks} were already
closed by Paper~A: (a) by Paper~A's no-splitting theorem
(Paper~A, Theorem~11) and (b) by the minimal-authority safe-default
definition (Paper~A, Definition~2) that fixes $a^0$ ex ante and
contractually. On this instance, the five-attack space is therefore
jointly closed by the Paper~A contract (a, b) together with four
Paper~C clauses: aggregate settlement (for c), the ambiguity-reserve
extension and escalation adjudication (jointly for d), and the
type-reporting penalty schedule $\kappa_r^\star$ (for e). The
worked example is presented in the \emph{premium-free baseline}
($P_\theta \equiv 0$ in clause (C6)); the premium-aware bound
$\widetilde{\kappa}_r^\star$ used by the general composition theorem
(Sections~\ref{sec:composition}--\ref{sec:ir-bb}) reduces to
$\kappa_r^\star$ in this baseline (cf.\ clause (C5) of
Definition~\ref{def:composition}).

The example is deliberately stylised: the boundary potentials are
linear, the type space is binary, and the verification protocol's
detection probability is treated as a single scalar.
Section~\ref{sec:composition} promotes the worked-instance observation
to a general statement: under any composition contract
(Definition~\ref{def:composition}) the five-attack space is jointly
closed (Theorem~\ref{thm:composition}), with general $\Phi$,
multi-type $\Theta$, and verification-protocol-dependent $q_{\theta r}$.

\section{Composition: Joint Incentive Compatibility}\label{sec:composition}

Theorems~\ref{thm:boundary-agg}, \ref{thm:interface}, and \ref{thm:type}
each close one attack class of Definition~\ref{def:attacks} in
isolation. Section~\ref{sec:example} demonstrated that the three
closures act simultaneously on a single worked instance. This section
formalises that observation. We bundle the Paper~A contract with the
four Paper~C clauses introduced above into a single \emph{composition
contract}, identify a uniformity property of the per-clause closures
(Lemma~\ref{lem:report-indep}), and prove the joint
incentive-compatibility result (Theorem~\ref{thm:composition}) over
the full five-attack space. One auxiliary hypothesis---operator
individual rationality---is treated as a black box and deferred to
Section~\ref{sec:obligations}.

\begin{definition}[Composition contract]\label{def:composition}
A \emph{composition contract} consists of the Paper~A contract together
with four Paper~C clauses:
\begin{enumerate}[label=\textup{(C\arabic*)},leftmargin=2.1em,itemsep=2pt]
  \item \textbf{Paper~A clauses.} The minimal-authority safe default
  $a^0$ fixed ex ante (Paper~A, Definition~2); the no-splitting
  potential structure (Paper~A, Theorem~11); and the high-probability
  toll envelope (Paper~A, Assumption~16) with the associated runtime
  gate.
  \item \textbf{Aggregate settlement.} A common-control aggregation
  map $\pi$ (Definition~\ref{def:cc}) and the aggregate-settlement
  rule (Definition~\ref{def:aggregate-settlement}) such that the
  operator's payable toll over $\Bset$ equals
  $\Phi_{B^\star}(E^{B^\star}_T) - \Phi_{B^\star}(E^{B^\star}_0)$.
  \item \textbf{Ambiguity-reserve extension.} A Paper~C augmentation
  of the runtime gate that reserves every unparseable or
  multiply-interpretable tool call at the supremum of the
  deterministic ambiguity set,
  $\sup_{a \in \mathcal{A}_{\mathrm{ambig}}} \mathrm{Reserve}(a)$,
  and treats the event as a priced interface failure (this clause is
  not in Paper~A; we propose it here as a Paper~C/AAI augmentation).
  \item \textbf{Escalation adjudication.} The adjudication rule
  $\xi_{\mathrm{escalate}}$ (Definition~\ref{def:if}) with escalation
  fee
  \[
  \kappa_{\mathrm{esc}}
  \;\ge\;
  \sup_{\theta \in \Theta,\, r \in \Theta}\bigl(\mu_{c, \theta, r} - V_\theta - C_{\mathrm{fail},\theta}\bigr),
  \]
  where
  $\mu_{c, \theta, r} = \E_\theta[c_t^+ \mid \text{contract evaluated under report } r]$
  is the expected successful-action positive toll when the deployed
  type is $\theta$ (which determines the intended-action distribution
  as in Theorem~\ref{thm:interface}) and the contract is evaluated
  under the report-$r$ toll potential. The bound therefore makes the
  Theorem~\ref{thm:interface}(b) sufficient condition hold uniformly
  over the pair (deployed type, declared report) $\in \Theta \times \Theta$
  without assuming type-invariance of the successful-action toll
  distribution.
  \item \textbf{Type-reporting penalty schedule.} A finite declarable
  type set $\Theta$ and verification protocol as in
  Definition~\ref{def:type-contract}, together with a penalty schedule
  $\{\kappa_r\}_{r \in \Theta}$ satisfying the \emph{premium-aware}
  bound
  \[
  \kappa_r \;\ge\; \widetilde{\kappa}_r^\star
  \;:=\; \max_{\theta \ne r:\, q_{\theta r} > 0}
  \frac{[\widetilde{\Delta}_{\theta r}]_+}{q_{\theta r}}
  \qquad \forall\, r \in \Theta,
  \]
  where the \emph{net misreport gain} accounts for the premium
  schedule of clause (C6) below,
  \[
  \widetilde{\Delta}_{\theta r}
  \;:=\;
  (U^0_{\theta r} - P_r) - (U^0_{\theta\theta} - P_\theta)
  \;=\;
  \Delta_{\theta r} + P_\theta - P_r,
  \]
  with $\Delta_{\theta r}$ and $q_{\theta r}$ as in
  Section~\ref{sec:type} and Lemma~\ref{lem:premium-ic} verifying that
  $\widetilde{\kappa}_r^\star$ is the relevant IC bound. In the
  premium-free baseline ($P_\theta \equiv 0$),
  $\widetilde{\Delta}_{\theta r} = \Delta_{\theta r}$ and
  $\widetilde{\kappa}_r^\star = \kappa_r^\star$, recovering
  Theorem~\ref{thm:type}.
  \item \textbf{Premium schedule.} A schedule
  $\{P_\theta\}_{\theta \in \Theta}$ of ex-ante premiums collected at
  $t = 0$ from operators declaring type $\theta$. The choice of
  $\{P_\theta\}$ is determined by individual rationality and
  budget balance (Section~\ref{sec:ir-bb},
  Theorem~\ref{thm:ir-bb}); for Theorem~\ref{thm:composition} it is
  treated as a contract parameter satisfying
  Assumption~\ref{ass:ir}.
\end{enumerate}
\end{definition}

\begin{assumption}[Operator individual rationality]\label{ass:ir}
For every operator type $\theta \in \Theta$ with premium $P_\theta$
and outside-option utility $\bar U_\theta$, the contract terms satisfy
\[
U^0_{\theta\theta} - P_\theta \;\ge\; \bar U_\theta.
\]
Constructing $(P_\theta, \Phi_{r,\cdot}, \kappa_r)$ that jointly
implement this inequality and budget balance is deferred to
Section~\ref{sec:obligations}; here we treat it as a hypothesis (see
Remark~\ref{rem:ir-budget}).
\end{assumption}

\begin{assumption}[Static-deployment scope]\label{ass:static-deploy}
For Theorem~\ref{thm:composition}, the operator's type space $\Theta$
is the finite model-identity projection of the richer operator type
space $\Oset$ of Section~\ref{sec:operator}: a type
$\theta \in \Theta$ identifies the operator with a single deployed
model identity over the entire underwriting horizon
$[0, T]$. All non-model coordinates of the Section~\ref{sec:operator}
type tuple
$(\M_\theta, a^0_\theta, B_\theta, A_\theta^+)$---the safe-default
menu choice, the boundary partition, and the maximum authority
set---are fixed by the contract menu of Definition~\ref{def:type-contract}
at $t = 0$ and are not separately strategic post-inception. Any
runtime switch of deployed model identity is treated as a fresh
report event and triggers the verification protocol of
Definition~\ref{def:type-contract}(iii) on the new identity.
\end{assumption}

\begin{assumption}[Premium-aware detectability]\label{ass:detect-net}
For every pair $(\theta, r) \in \Theta \times \Theta$ with $r \neq \theta$, if the net misreport gain
\[
\widetilde{\Delta}_{\theta r}
\;:=\;
(U^0_{\theta r} - P_r) - (U^0_{\theta\theta} - P_\theta)
\;=\;
\Delta_{\theta r} + P_\theta - P_r
\]
is strictly positive, then the verification protocol's detection
probability satisfies $q_{\theta r} > 0$.
\end{assumption}

Assumption~\ref{ass:detect-net} strengthens
Assumption~\ref{ass:detect-profitable} (Section~\ref{sec:type}): the
latter requires detectability of \emph{gross}-profitable misreports
($\Delta_{\theta r} > 0$), whereas
Assumption~\ref{ass:detect-net} additionally requires detectability of
\emph{net}-profitable misreports motivated by a premium discount
($P_\theta - P_r$ pushing $\widetilde{\Delta}_{\theta r}$ above zero
even when $\Delta_{\theta r} \le 0$). In the premium-free baseline
($P_\theta \equiv 0$) the two coincide.

\paragraph{Reduction strategy.}
The composition argument rests on a uniformity observation: clauses
(C1)--(C4) close attack classes (a)--(d) regardless of which type
$r \in \Theta$ the operator reports. This collapses the operator's
optimisation to a type-reporting problem, to which the premium-aware
type-reporting analysis of Lemma~\ref{lem:premium-ic} (introduced in
Section~\ref{sec:ir-bb}) applies directly. The next lemma formalises
the uniformity.

\begin{lemma}[Report-independent closure of attack classes (a)--(d)]\label{lem:report-indep}
Fix a composition contract (Definition~\ref{def:composition}). For
every operator true type $\theta \in \Theta$ and every declared report
$r \in \Theta$:
\begin{enumerate}[label=\textup{(\alph*)},leftmargin=1.8em,itemsep=2pt]
  \item \emph{Within-boundary splitting.} The cumulative toll on any
  single sub-boundary $B_k$ depends on the executed sequence only via
  the potential identity
  $\Phi_{r, B_k}(E^{B_k}_T) - \Phi_{r, B_k}(E^{B_k}_0)$; no decomposition
  into smaller priced increments inside $B_k$ reduces this quantity.
  \item \emph{Post-toll safe-default selection.} The safe-default
  mapping $a^0_r$ is contractually fixed at $t = 0$; no realisation of
  the toll process changes it.
  \item \emph{Cross-boundary re-routing.} No partition
  $\{e_k\}_{k=1}^K$ of a fixed economic increment $E^\star$ across
  $\Bset$ reduces payable toll below
  $\Phi_{r, B^\star}(E^\star) - \Phi_{r, B^\star}(0)$.
  \item \emph{Interface-compliance gaming.} For every realisation of
  the operator's interface-failure rate $f_\theta$, the operator's
  expected utility under report $r$ is non-increasing in $f_\theta$.
\end{enumerate}
Each of (a)--(d) holds for all pairs $(\theta, r) \in \Theta \times \Theta$;
the closures are uniform in the operator's reporting strategy.
\end{lemma}

\begin{proof}[Proof of Lemma~\ref{lem:report-indep}]
(a) Paper~A's no-splitting theorem (Paper~A, Theorem~11) shows that
within a single underwriting boundary $B$ with cumulative exposure
$E^B_t$ and potential $\Phi_B$, the executed-action toll satisfies
$\sum_t c_t^+ = \Phi_B(E^B_T) - \Phi_B(E^B_0)$ for any decomposition of
the priced actions into increments. Composition contract clause (C1)
adopts this potential structure on each sub-boundary $B_k$ of the
operator's partition $\Bset$. Under (C5) the contract instantiates the
sub-boundary potential as $\Phi_{r, B_k}$ for the declared report $r$,
but Paper~A's argument is parametric in the potential and does not
use which $r$ was reported. Closure of (a) therefore holds uniformly in
$(\theta, r)$.

(b) Paper~A's Definition~2 fixes the safe-default mapping $a^0$ as a
contractual object known at $t = 0$. Composition contract clause~(C5)
combined with Definition~\ref{def:type-contract}(ii) selects
$a^0_r$ from the declared menu at $t = 0$. Once selected, $a^0_r$ is
unobservable as a function of subsequent toll realisations; no
operator deviation after $t = 0$ can change the safe default. The
report $r$ enters only as a menu-selection index at inception and does
not affect the closure: post-$t=0$ safe-default reselection is
unavailable for every $r$.

(c) Theorem~\ref{thm:boundary-agg} (steps~1--5 of its proof)
demonstrates that, given (C2) (common-control aggregation $\pi$ and
aggregate settlement), the payable toll on any partition
$\{e_k\}_{k=1}^K$ with $\sum_k e_k = E^\star$ equals
$\Phi_{B^\star}(E^\star) - \Phi_{B^\star}(0)$. The argument is
parametric in $\Phi_{B^\star}$: under composition contract the
aggregate potential is $\Phi_{r, B^\star}$, but neither the
common-control observability requirement (Definition~\ref{def:cc}(iii))
nor the aggregate-settlement identity uses the report $r$ in a
strategic role. Closure of (c) holds uniformly in $(\theta, r)$.

(d) Composition contract clause~(C3) reserves any ambiguous tool call
at $\sup_{a \in \mathcal{A}_{\mathrm{ambig}}} \mathrm{Reserve}(a)$;
this supremum is over the contract's deterministic parsing
specification and is independent of $(\theta, r)$. Composition contract
clause~(C4) imposes the escalation fee
$\kappa_{\mathrm{esc}} \ge \sup_{\theta', r'}(\mu_{c, \theta', r'} - V_{\theta'} - C_{\mathrm{fail}, \theta'})$,
which by construction dominates the Theorem~\ref{thm:interface}(b)
sufficient condition for every (deployed type, report) pair
$(\theta, r) \in \Theta \times \Theta$: at any such pair, the
type-and-report-specific expected successful-action toll
$\mu_{c, \theta, r}$ enters the (C4) supremum, so
$\kappa_{\mathrm{esc}} \ge \mu_{c, \theta, r} - V_\theta - C_{\mathrm{fail}, \theta}$
holds. Applying Theorem~\ref{thm:interface}(b) at $(\theta, r)$,
the operator's expected utility is non-increasing in the operator-induced
$f_\theta$ regardless of which report was declared. Closure of (d)
holds uniformly in $(\theta, r)$.
\end{proof}

\paragraph{From reduction to joint IC.}
With (a)--(d) uniformly closed across reports, the operator's optimal
execution given any report collapses to contract-following choices;
under Assumption~\ref{ass:static-deploy} the deployed type is not a
strategy variable, so the only remaining strategic lever is the report
itself. The premium-aware type-reporting penalty schedule (C5) and
Lemma~\ref{lem:premium-ic} then close~(e).

\begin{theorem}[Joint IC of Paper~A + Paper~C]\label{thm:composition}
Fix a composition contract (Definition~\ref{def:composition}) and
suppose Assumption~\ref{ass:static-deploy} (static-deployment scope),
Assumption~\ref{ass:ir} (operator IR), and
Assumption~\ref{ass:detect-net} (premium-aware detectability) hold,
with clause (C5) instantiated at the premium-aware bound
$\kappa_r \ge \widetilde{\kappa}_r^\star$ for every $r \in \Theta$.
Then \emph{truthful contract-following}---declaring
the operator's true type $\hat\theta = \theta$, deploying the model
associated with $\theta$, following the safe-default mapping $a^0_\theta$
in the runtime gate, executing only priced actions that do not split,
re-route, or generate interface failures, and producing well-formed
tool calls---is weakly dominant for the operator over the entire
five-attack space of Definition~\ref{def:attacks}.
\end{theorem}

\begin{proof}[Proof of Theorem~\ref{thm:composition}]
The proof has four steps. Step~1 fixes the operator's strategy space
and writes utility in a form that separates report from execution.
Step~2 applies Lemma~\ref{lem:report-indep} to collapse the
execution-side optimisation conditional on report. Step~3 reduces
the remaining problem to premium-aware type reporting. Step~4 applies
Lemma~\ref{lem:premium-ic}. Throughout, $\theta \in \Theta$ denotes the
operator's true type. Under
Assumption~\ref{ass:static-deploy}, $\theta$ is identified with the
deployed model identity over the entire underwriting horizon and the
non-model coordinates of the Section~\ref{sec:operator} type tuple
are fixed by contract menu at $t = 0$; the operator's strategic
choices reduce to what to report and how to execute. Attack-class (e)
of Definition~\ref{def:attacks} is then the choice $r \neq \theta$
at contract inception, while attack-classes (a)--(d) are choices of
execution policy.

\medskip
\noindent\emph{Step 1 (Strategy decomposition and utility).}
With the deployed type identified with $\theta$, an operator strategy
is a pair $(\rho, \sigma_e)$ where $\rho \in \Delta(\Theta)$ is a
report distribution and $\sigma_e$ is an execution policy over priced
actions and interface behaviour (including the operator's induced
interface-failure rate $f_\theta$, which $\sigma_e$ can raise via
ambiguous tool-call construction). Writing
$\mathrm{TollNet}_{\theta, r}(\sigma_e)$ for the cumulative payable
toll under report $r$, deployed type $\theta$, and execution
$\sigma_e$ (including ambiguity and escalation charges under
(C3)--(C4)), $U^{\mathrm{biz}}_\theta(\sigma_e \mid r)$ for the gross
deployment value, and $P_r$ for the ex-ante premium charged by
clause~(C6) under report $r$, the operator's expected utility is
\[
U_\theta(\rho, \sigma_e)
\;=\;
\E_{r \sim \rho}\!\left[\;
U^{\mathrm{biz}}_\theta(\sigma_e \mid r)
- \mathrm{TollNet}_{\theta, r}(\sigma_e)
- P_r
- q_{\theta, r}\,\kappa_r
\;\right],
\]
where $q_{\theta, r}\,\kappa_r$ is the expected verification penalty
under report $r$ ($q_{\theta\theta} = 0$ by convention). Deployment is
not a strategy variable: the operator's type $\theta$ is fixed, and
attack-class (e) (model-identity misreporting) is the gap between
$\theta$ and $r$, not a deployment deviation. The ex-ante premium
$P_r$ appears in the utility because it is charged on the operator's
declared type, not their deployed type.

\medskip
\noindent\emph{Step 2 (Collapse of $\sigma_e$ conditional on report).}
Fix any report $r$ and the operator's true type $\theta$.
Lemma~\ref{lem:report-indep} shows that within the
conditional-on-$(\theta, r)$ problem, each of attack classes (a)--(d)
is weakly utility-non-positive: an execution $\sigma_e$ that splits
priced actions inside any $B_k$, attempts cross-boundary re-routing of
an aggregate exposure increment, or raises the interface-failure
rate $f_\theta$ via ambiguous tool calls can be replaced by a
contract-following execution $\sigma_e^\star(\theta, r)$ without
decreasing expected utility. Define the \emph{report-conditional
best-response utility}
\[
U^0_{\theta, r}
\;:=\;
\max_{\sigma_e}\,
\E_\theta\!\left[
U^{\mathrm{biz}}_\theta(\sigma_e \mid r)
- \mathrm{TollNet}_{\theta, r}(\sigma_e)
\right],
\]
where the maximum is attained on $\sigma_e^\star(\theta, r)$ by the
lemma. The maximum is over execution only; the deployed type
$\theta$ is held fixed. This $U^0_{\theta, r}$ coincides with the
construct of Section~\ref{sec:type}: the operator's pre-penalty
utility under report $r$ and deployed type $\theta$, with gaming
routes (a)--(d) closed. The verification posterior
$q_{\theta, r}$ and gross gain $\Delta_{\theta, r}$ of
Section~\ref{sec:type} are therefore defined relative to the same
post-(a)--(d) utility used here, so substituting $U^0_{\theta, r}$ into
the type-reporting analysis is consistent.

\medskip
\noindent\emph{Step 3 (Reduction to premium-aware type-reporting).}
Substituting $U^0_{\theta, r}$ into Step~1's utility expression, the
operator's problem reduces to
\[
\max_\rho\, \E_{r \sim \rho}\!\left[\,U^0_{\theta, r} - P_r - q_{\theta, r}\,\kappa_r\,\right]
\;=\;
\max_\rho\, \E_{r \sim \rho}\!\left[\,U_{\theta, r}\,\right],
\]
with $U_{\theta, r} = (U^0_{\theta, r} - P_r) - q_{\theta, r}\,\kappa_r$.
By linearity in $\rho$, the maximum is attained on some deterministic
report; the problem is exactly the \emph{premium-aware}
type-reporting problem of
Lemma~\ref{lem:premium-ic} (Section~\ref{sec:ir-bb}).

\medskip
\noindent\emph{Step 4 (Apply Lemma~\ref{lem:premium-ic}).}
Composition contract clause~(C5) imposes
$\kappa_r \ge \widetilde{\kappa}_r^\star$ for every $r \in \Theta$
(the premium-aware bound). Under
Assumption~\ref{ass:detect-net}, Lemma~\ref{lem:premium-ic} then
gives $U_{\theta, \theta} \ge U_{\theta, r}$ for every
$\theta \in \Theta$ and $r \in \Theta$. Truthful reporting
$\rho_{\rm tru}(\theta) = \delta_\theta$ is therefore weakly dominant.
In the premium-free baseline ($P_\theta \equiv 0$), this reduces to
Theorem~\ref{thm:type} applied directly under
Assumption~\ref{ass:detect-profitable}, since
$\widetilde{\kappa}_r^\star = \kappa_r^\star$ and
$\widetilde{\Delta}_{\theta r} = \Delta_{\theta r}$.

\medskip
\noindent\emph{Combining.}
Step~4 selects $r = \theta$, closing attack-class (e). Conditional on
this report, Step~2's lemma collapses $\sigma_e$ onto contract-following
$\sigma_e^\star(\theta, \theta)$, which respects $a^0_\theta$ on covered
events and avoids splitting, re-routing, and ambiguous tool calls
(closing (a)--(d)). The composition of these choices---the operator's
deployed type $\theta$ together with truthful report and contract-following
execution---is truthful contract-following;
Assumption~\ref{ass:ir} guarantees that this strategy weakly dominates
the outside option. The five-attack space of
Definition~\ref{def:attacks} is therefore jointly closed, and truthful
contract-following is weakly dominant.
\end{proof}

\begin{remark}[Coverage of the strategy space]\label{rem:coverage}
Theorem~\ref{thm:composition} is a claim relative to the five-attack
space of Definition~\ref{def:attacks}. A strategy that fits none of
(a)--(e)---for example, adversarial calibration of the toll envelope
(Section~\ref{sec:obligations}, item~2)---is outside the theorem's
scope. Joint IC over a finer attack characterisation requires the
corresponding contract clause and is not asserted here.
\end{remark}

\begin{remark}[IR and budget balance as deferred work]\label{rem:ir-deferred}
Assumption~\ref{ass:ir} is a hypothesis, not a conclusion. A
constructive lemma giving conditions on premiums $P_\theta$, toll
potentials, and the penalty schedule $\{\kappa_r\}_{r \in \Theta}$ under
which IR holds jointly with budget balance---e.g.\ by rebating expected
penalty surplus
$\sum_\theta \Pr[\hat\theta = \theta]\,\E[\,q\,\kappa\,]$ to truthful
types---is the next refinement of Theorem~\ref{thm:composition} and is
the first item of Section~\ref{sec:obligations}. The qualitative
discussion in Remark~\ref{rem:ir-budget} previews the relationship
between IR and the IC inequalities defining $\kappa_r^\star$.
\end{remark}

\begin{remark}[Composition versus welfare]\label{rem:welfare}
Theorem~\ref{thm:composition} establishes joint IC for one specific
parameterisation: conservative toll potentials, $\xi_{\mathrm{escalate}}$
in preference to $\xi_{\mathrm{safe}}$, and componentwise minimal
premium-aware penalties $\widetilde{\kappa}_r^\star$. An alternative
mechanism designer optimising ex-ante operator welfare under the same
gaming constraints might choose different parameters (e.g.\ a smaller
$\kappa_{\mathrm{esc}}$ when business value $V_\theta$ is uniformly
high). Welfare characterisation under the joint IC constraint is left
for future work.
\end{remark}

\section{Individual Rationality and Budget Balance}\label{sec:ir-bb}

Theorem~\ref{thm:composition} treats operator individual rationality
(Assumption~\ref{ass:ir}) as a hypothesis. This section shows that
the hypothesis can be discharged constructively: under mild positivity
conditions, there is an explicit two-parameter family of premium
schedules $\{P_\theta\}_{\theta \in \Theta}$ that jointly satisfies IR
and weak budget balance at the Theorem~\ref{thm:composition}
equilibrium, parametrised by how aggressively the contract extracts
operator surplus (load) versus rebates expected toll surplus
(rebate). Three natural special cases give exact budget balance,
pure-load extraction, and a mixed contract; combining the
constructive schedule with Theorem~\ref{thm:composition} yields a
single statement (Theorem~\ref{thm:full-composition}) of joint
incentive compatibility, individual rationality, and weak budget
balance.

\begin{definition}[IR slack and toll surplus]\label{def:slack}
For each operator type $\theta \in \Theta$, define
\begin{align*}
S_\theta \;&:=\; U^0_{\theta\theta} - \bar U_\theta,
\qquad &&\text{(\emph{IR slack})}\\
\Delta_\theta \;&:=\; \E_\theta\!\left[\mathrm{TollNet}_{\theta,\theta}\right] - \mathrm{Cost}_\theta,
\qquad &&\text{(\emph{toll surplus})}
\end{align*}
where $U^0_{\theta\theta}$ and $\mathrm{TollNet}_{\theta,\theta}$ are
the post-(a)--(d) pre-penalty utility and cumulative payable toll of
Section~\ref{sec:composition}, $\bar U_\theta$ is the operator's
outside-option utility, and $\mathrm{Cost}_\theta$ is the
underwriter's expected loss-cum-overhead cost under deployed type
$\theta$ and truthful contract-following.
\end{definition}

\begin{assumption}[Positive IR slack and toll surplus]\label{ass:positivity}
$S_\theta \ge 0$ and $\Delta_\theta \ge 0$ for every $\theta \in \Theta$.
\end{assumption}

$S_\theta \ge 0$ is the operator-side condition that joining the
contract is at least as good as the outside option before any ex-ante
premium is charged. $\Delta_\theta \ge 0$ is the underwriter-side
condition that the expected payable toll under truth covers the
expected cost. The latter is treated as a separate actuarial loading
assumption: Paper~A's conservative-envelope guarantee
(Paper~A, Assumption~16) gives high-probability dominance of toll over
realised loss within the calibration event of probability at least
$1 - \alpha$, but expected dominance
$\E_\theta[\mathrm{TollNet}_{\theta,\theta}] \ge \mathrm{Cost}_\theta$
requires additional control over the residual probability-$\alpha$
tail. Standard actuarial constructions---loading factor on the
expected loss, capped excess layer, or reinsurance treaty---supply
this dominance; we take $\Delta_\theta \ge 0$ as the contract-design
input rather than derive it from Paper~A's envelope alone.

The premium schedule $\{P_\theta\}$ enters the operator's
type-reporting incentive because a deviation from true type $\theta$
to report $r$ switches the ex-ante premium from $P_\theta$ to $P_r$.
The relevant IC bound for the composition contract is therefore the
\emph{premium-aware} minimum penalty $\widetilde{\kappa}_r^\star$
introduced in clause (C5) of Definition~\ref{def:composition}; the
following lemma verifies that the bound is sufficient and tight.

\begin{lemma}[Premium-aware type-reporting IC]\label{lem:premium-ic}
Fix a composition contract (Definition~\ref{def:composition}) with
premium schedule $\{P_\theta\}$ in clause (C6). Under
Assumption~\ref{ass:detect-net} (premium-aware detectability),
truthful type-reporting is weakly dominant for the operator if and
only if, for every $r \in \Theta$,
\[
\kappa_r \;\ge\; \widetilde{\kappa}_r^\star
\;=\; \max_{\theta \ne r:\, q_{\theta r} > 0}
\frac{[\widetilde{\Delta}_{\theta r}]_+}{q_{\theta r}}.
\]
Moreover, $\widetilde{\kappa}_r^\star$ is the componentwise minimal
penalty schedule that implements truthfulness given the premium
schedule.
\end{lemma}

\begin{proof}[Proof of Lemma~\ref{lem:premium-ic}]
The argument follows Theorem~\ref{thm:type} (Steps~1--5 of its
proof), with one substitution: the operator's
post-premium pre-penalty utility under report $r$ is
$U^0_{\theta r} - P_r$ rather than $U^0_{\theta r}$, since the
premium $P_r$ is collected at $t = 0$ from operators declaring
$\hat\theta = r$.

For $r \ne \theta$, the operator's pairwise IC condition then becomes
\[
(U^0_{\theta\theta} - P_\theta) \;\ge\; (U^0_{\theta r} - P_r) - q_{\theta r}\,\kappa_r,
\]
which rearranges to
\[
q_{\theta r}\,\kappa_r \;\ge\; (U^0_{\theta r} - P_r) - (U^0_{\theta\theta} - P_\theta) \;=\; \widetilde{\Delta}_{\theta r}.
\]
The case split of Theorem~\ref{thm:type}'s Step~2 carries over with
$\widetilde{\Delta}_{\theta r}$ replacing $\Delta_{\theta r}$:
Case~A ($\widetilde{\Delta}_{\theta r} \le 0$) gives the inequality
automatically since the left-hand side is non-negative; Case~B
($\widetilde{\Delta}_{\theta r} > 0$) requires the right-hand side
divided by $q_{\theta r}$ to be a finite bound, which requires
$q_{\theta r} > 0$. The premium-aware detectability assumption
(Assumption~\ref{ass:detect-net}) supplies this: it asserts
$\widetilde{\Delta}_{\theta r} > 0 \Rightarrow q_{\theta r} > 0$,
covering exactly the case at hand. (The original
Assumption~\ref{ass:detect-profitable} of
Section~\ref{sec:type} would not suffice when $P_\theta - P_r > 0$
makes $\widetilde{\Delta}_{\theta r} > 0$ while $\Delta_{\theta r} \le 0$,
since the original assumption is stated for gross $\Delta_{\theta r}$
only.) The pairwise IC then yields
$\kappa_r \ge [\widetilde{\Delta}_{\theta r}]_+ / q_{\theta r}$;
taking the maximum over $\theta \ne r$ recovers
$\widetilde{\kappa}_r^\star$. Sufficiency and necessity follow as in
Theorem~\ref{thm:type}'s Steps~3--5, and componentwise minimality
follows from Step~5: any $\kappa_r < \widetilde{\kappa}_r^\star$ fails
the pairwise IC for the binding $(\theta, r)$ pair attaining the
maximum.

In the premium-free baseline ($P_\theta \equiv 0$),
$\widetilde{\Delta}_{\theta r} = \Delta_{\theta r}$ and
Assumption~\ref{ass:detect-net} reduces to
Assumption~\ref{ass:detect-profitable}; the result then reduces to
Theorem~\ref{thm:type}.
\end{proof}

\begin{definition}[IR-feasibility and weak budget balance]\label{def:ir-bb}
Fix a composition contract (Definition~\ref{def:composition}) with
premium schedule $\{P_\theta\}_{\theta \in \Theta}$ and a prior
$\mu \in \Delta(\Theta)$ over types. The schedule is:
\begin{itemize}[leftmargin=1.5em,itemsep=2pt]
  \item \emph{IR-feasible} if $U^0_{\theta\theta} - P_\theta \ge \bar U_\theta$
  for every $\theta \in \Theta$, equivalently $P_\theta \le S_\theta$.
  \item \emph{Weakly budget-balanced} at the
  Theorem~\ref{thm:composition} equilibrium if the underwriter's
  expected net revenue from truthful operators is non-negative:
  \[
  \sum_{\theta \in \Theta} \mu(\theta)\,\bigl[P_\theta + \Delta_\theta\bigr] \;\ge\; 0.
  \]
\end{itemize}
Penalty revenue from $\kappa_r$ is zero at the truthful equilibrium of
Theorem~\ref{thm:composition} ($q_{\theta\theta} = 0$) and therefore
does not appear in the budget identity; off-equilibrium-path penalty
revenue is treated in Remark~\ref{rem:off-eq} below.
\end{definition}

\begin{theorem}[Constructive IR + BB premium schedule]\label{thm:ir-bb}
Suppose Assumption~\ref{ass:positivity} holds. Let
$\bar S := \E_\mu[S_\theta]$ and $\bar \Delta := \E_\mu[\Delta_\theta]$
denote the prior expectations of the IR slack and toll surplus. For
any pair $(\alpha, \beta) \in [0, 1] \times \R_+$ satisfying the
\emph{budget-balance feasibility condition}
\begin{equation}
\alpha\,\bar S + (1 - \beta)\,\bar \Delta \;\ge\; 0,
\label{eq:bb-feas}
\end{equation}
the premium schedule
\begin{equation}
P_\theta \;=\; \alpha\,S_\theta - \beta\,\Delta_\theta
\qquad \forall\, \theta \in \Theta
\label{eq:premium-schedule}
\end{equation}
is IR-feasible and weakly budget-balanced.
\end{theorem}

\begin{proof}[Proof of Theorem~\ref{thm:ir-bb}]
The proof has two steps.

\medskip
\noindent\emph{Step 1 (IR-feasibility).}
Fix any $\theta \in \Theta$. Using
$U^0_{\theta\theta} = S_\theta + \bar U_\theta$ and the definition of
$P_\theta$,
\begin{align*}
U^0_{\theta\theta} - P_\theta
&= S_\theta + \bar U_\theta - \alpha\,S_\theta + \beta\,\Delta_\theta\\
&= (1 - \alpha)\,S_\theta + \beta\,\Delta_\theta + \bar U_\theta.
\end{align*}
By Assumption~\ref{ass:positivity} both $S_\theta \ge 0$ and
$\Delta_\theta \ge 0$; by hypothesis $1 - \alpha \ge 0$ and
$\beta \ge 0$. Hence $(1 - \alpha) S_\theta + \beta \Delta_\theta \ge 0$,
so $U^0_{\theta\theta} - P_\theta \ge \bar U_\theta$. IR-feasibility
holds for every $\theta$.

\medskip
\noindent\emph{Step 2 (Weak budget balance).}
Substituting the schedule~\eqref{eq:premium-schedule} into the BB
identity of Definition~\ref{def:ir-bb},
\begin{align*}
\sum_{\theta} \mu(\theta)\,\bigl[P_\theta + \Delta_\theta\bigr]
&= \sum_{\theta} \mu(\theta)\,\bigl[\alpha S_\theta - \beta \Delta_\theta + \Delta_\theta\bigr]\\
&= \alpha \sum_\theta \mu(\theta) S_\theta + (1 - \beta) \sum_\theta \mu(\theta) \Delta_\theta\\
&= \alpha\,\bar S + (1 - \beta)\,\bar \Delta.
\end{align*}
By hypothesis~\eqref{eq:bb-feas} this quantity is non-negative; weak
budget balance holds.
\end{proof}

\begin{corollary}[Three natural parameterisations]\label{cor:ir-bb-params}
Under Assumption~\ref{ass:positivity}, each of the following
$(\alpha, \beta)$ choices satisfies~\eqref{eq:bb-feas} without further
constraint, and the corresponding $P_\theta$ is IR-feasible and weakly
budget-balanced:
\begin{enumerate}[label=\textup{(\alph*)},leftmargin=1.8em,itemsep=2pt]
  \item \emph{Pure-load contract} $(\alpha = 1, \beta = 0)$:
  $P_\theta = S_\theta$. The contract extracts the operator's full IR
  slack as ex-ante load; truthful participation is IR with equality.
  Budget surplus is $\bar S + \bar \Delta \ge 0$.
  \item \emph{Toll-rebate contract} $(\alpha = 0, \beta = 1)$:
  $P_\theta = -\Delta_\theta$. The contract rebates the expected toll
  surplus to the operator at $t = 0$ as a negative premium. Budget
  identity is exact: $\bar S \cdot 0 + \bar \Delta \cdot 0 = 0$. IR
  holds strictly whenever $S_\theta > 0$ for some $\theta$, with
  slack equal to $S_\theta + \Delta_\theta$.
  \item \emph{Mixed contract} $(\alpha = 1, \beta = 1)$:
  $P_\theta = S_\theta - \Delta_\theta$. The contract extracts the IR
  slack as load while rebating the expected toll surplus; the two
  effects partially cancel. IR holds with strict slack $\Delta_\theta$;
  budget surplus is $\bar S \ge 0$.
\end{enumerate}
These are three representative parameterisations of the
feasibility region $\{(\alpha, \beta) \in [0, 1] \times \R_+ :
\alpha \bar S + (1 - \beta) \bar \Delta \ge 0\}$, not an exhaustive
Pareto characterisation: (a) maximises the underwriter's expected
revenue subject to IR; (b) gives pointwise zero net revenue per type
($P_\theta + \Delta_\theta = 0$); (c) extracts IR slack while
rebating expected toll. Interior points $(\alpha, \beta) \in (0, 1)
\times (0, 1)$ and other feasible boundary points may matter for
specific welfare or fairness objectives; a formal Pareto
characterisation under a stated welfare order is left as future
work.
\end{corollary}

\begin{theorem}[Joint IC + IR + BB composition]\label{thm:full-composition}
Fix a composition contract (Definition~\ref{def:composition}), with
the premium schedule of clause (C6) chosen as in
Theorem~\ref{thm:ir-bb} for some
$(\alpha, \beta) \in [0,1] \times \R_+$ satisfying~\eqref{eq:bb-feas},
and the penalty schedule of clause (C5) satisfying the
\emph{premium-aware} bound
$\kappa_r \ge \widetilde{\kappa}_r^\star$ for every $r \in \Theta$
(with $\widetilde{\kappa}_r^\star$ as in clause (C5) and verified by
Lemma~\ref{lem:premium-ic}). Assume
Assumption~\ref{ass:static-deploy} (static deployment),
Assumption~\ref{ass:detect-net} (premium-aware detectability of
profitable misreports), and Assumption~\ref{ass:positivity} (positive
IR slack and toll surplus). Then:
\begin{enumerate}[label=\textup{(\roman*)},leftmargin=1.8em,itemsep=2pt]
  \item Truthful contract-following is weakly dominant for the
  operator across the five-attack space of
  Definition~\ref{def:attacks};
  \item Truthful participation is individually rational for every
  $\theta \in \Theta$;
  \item The mechanism is weakly budget-balanced at the truthful
  equilibrium.
\end{enumerate}
\end{theorem}

\begin{proof}[Proof of Theorem~\ref{thm:full-composition}]
(i) restates Theorem~\ref{thm:composition} with the (C5) bound
strengthened to the premium-aware version: attack classes (a)--(d)
remain closed by Lemma~\ref{lem:report-indep} (which is parametric in
the contract parameters and unaffected by the premium schedule), and
attack class (e) is closed by Lemma~\ref{lem:premium-ic} under
$\kappa_r \ge \widetilde{\kappa}_r^\star$, which is the premium-aware
analogue of Theorem~\ref{thm:type} used in Theorem~\ref{thm:composition}'s
Step~4. The composition argument of Steps~1--4 of
Theorem~\ref{thm:composition}'s proof carries over verbatim with
$\widetilde{\kappa}_r^\star$ replacing $\kappa_r^\star$ and
post-premium pre-penalty utilities $U^0_{\theta r} - P_r$ replacing
$U^0_{\theta r}$.

(ii) and (iii) are Theorem~\ref{thm:ir-bb} applied to the chosen
$(\alpha, \beta)$: IR-feasibility of $P_\theta = \alpha S_\theta - \beta
\Delta_\theta$ gives (ii); weak budget balance at the truthful
equilibrium under condition~\eqref{eq:bb-feas} gives (iii).

The three results hold simultaneously because Lemma~\ref{lem:report-indep}'s
closure of (a)--(d) is independent of the premium schedule and
penalty schedule, while Lemma~\ref{lem:premium-ic} handles (e) under
the premium-aware bound; the bound is met by construction since (C5)
is required to satisfy $\kappa_r \ge \widetilde{\kappa}_r^\star$ for
the premium schedule of (C6).
\end{proof}

\begin{remark}[Off-equilibrium penalty revenue]\label{rem:off-eq}
At the truthful equilibrium of Theorem~\ref{thm:composition}, no
operator misreports and $q_{\theta\theta} = 0$, so the penalty
revenue $\sum_\theta \mu(\theta)\,q_{\theta r}\,\kappa_r$ is zero
in expectation. Under noisy reporting models---e.g.\ a
trembling-hand perturbation in which truthful reporting is randomised
with probability-$\epsilon$ misreport noise---expected penalty revenue
is positive and could in principle be rebated to truthful types via
an ex-post refund $R_\theta \ge 0$. Such a rebate is an
\emph{informal extension} of Theorem~\ref{thm:ir-bb} and is not
automatic: any rebate rule that depends on the operator's report or
execution path perturbs the pairwise IC condition of
Lemma~\ref{lem:premium-ic}, so the rebate rule must be either
report-and-execution independent (e.g.\ a uniform rebate across all
declared types) or conditional only on verified truthfulness (paid
ex-post once the verification protocol clears the declared type).
A formal analysis under either restriction, including its effect on
the BB feasibility region, is left for future work
(Section~\ref{sec:obligations}, item~1).
\end{remark}

\begin{remark}[Welfare and the operator's reservation type]\label{rem:welfare-tradeoff}
The two-parameter family~\eqref{eq:premium-schedule} exposes the
mechanism designer's welfare trade-off. The toll-rebate contract
($\alpha = 0$) is operator-optimal among IR-feasible schedules; the
pure-load contract ($\alpha = 1$) is underwriter-optimal among
weakly-BB schedules. Intermediate $(\alpha, \beta)$ choices
correspond to bargaining points within the feasibility region of
Corollary~\ref{cor:ir-bb-params}. The result is silent on which point
to choose: that depends on principal preferences, competitive
constraints, and considerations outside the
Theorem~\ref{thm:composition} IC scope. A formal Pareto
characterisation under a stated welfare order is left for future
work.
\end{remark}

\section{Related Work}\label{sec:related}

\paragraph{Mechanism design and strategy-proofness.}
The contract-design questions here are classical mechanism design. The
Vickrey--Clarke--Groves mechanism \citep{vickrey1961counterspeculation,
clarke1971multipart, groves1973incentives} achieves dominant-strategy
truthfulness for quasilinear agents but requires budget-breaking
transfers; the Gibbard--Satterthwaite theorem
\citep{gibbard1973manipulation, satterthwaite1975strategy} shows that
strategy-proofness over unrestricted preferences forces dictatorship; and
Myerson--Satterthwaite \citep{myerson1983efficient} shows that no
mechanism is simultaneously efficient, individually rational, and
budget-balanced. We seek neither efficiency nor VCG-style transfers: the
runtime contract operates under a \emph{fixed} budget $B_0$ inherited from
Paper~A, and we ask only for incentive compatibility against a catalogued
attack space, which is why these impossibility results do not bind. The
penalty schedule of Theorem~\ref{thm:type} is a Myerson-style
\citep{myerson1981} virtual-cost construction specialised to a discrete
menu of model identities.

\paragraph{Contract theory and principal--agent models.}
The operator--underwriter relationship is a principal--agent problem with
both hidden action (which model the operator deploys) and hidden type
(which identity it reports), combining the moral-hazard frictions of
\citet{holmstrom1979moral} with the screening tradition
\citep{laffont2002theory}. Our contribution is to instantiate these
frictions in a \emph{runtime} setting where the action is a stream of
priced tool calls and the screening device is a menu of identity-indexed
safe defaults rather than a monetary transfer schedule.

\paragraph{AI agent control, guardrails, and safety.}
A large body of work governs the \emph{agent}: reinforcement learning
from human feedback \citep{ouyang2022rlhf}, constitutional methods
\citep{bai2022constitutional}, programmable guardrails
\citep{nvidia2023nemo}, capability-based access control
\citep{miller2006capabilities}, and threat catalogues such as the OWASP
LLM Top~10 \citep{owasp2024llm}. These shape what the agent does; they do
not price what a \emph{strategic operator} is incentivised to deploy.
The present paper is complementary: it treats interface reliability and model
identity as contractual variables and asks when the contract, not the
agent, is gaming-resistant. The interface-compliance result
(Theorem~\ref{thm:interface}) makes the relationship concrete: a
guardrail-style ``invalid output is a safe outcome'' reflex is exactly the
adjudication rule that creates a perverse operator incentive.

\paragraph{AI risk quantification, insurance, and evaluation.}
Closest in spirit is recent work quantifying financial risk for AI agents
\citep{hua2026quantifying} and algorithmic / data-driven insurance
\citep{bertsimas2021algorithmic}, alongside the established cyber-risk
insurance literature \citep{eling2016cyber}. These price or measure risk
but do not characterise operator gaming under the resulting contract. On
the measurement side, agent benchmarks such as AgentBench
\citep{liu2023agentbench}, $\tau$-bench \citep{yao2024taubench}, and
SWE-bench \citep{jimenez2024swebench} evaluate capability; the companion
Paper~B \citep{chen2026insuring} repurposes such environments for
actuarial measurement, and the present paper supplies the contract theory
that makes those measurements incentive-relevant. The program rests on the
time-consistent counterfactual runtime of Paper~A
\citep{chen2026runtime}.

\section{Conclusion}\label{sec:conclusion}

Paper~A specifies a runtime actuarial contract that prices each
side-effect-bearing action against a safe default under a conservative
budget envelope, but treats the operator as a passive deployment. This
paper makes the operator strategic and asks where that contract is
already gaming-resistant and where it needs additional clauses. We
catalogue a five-attack space (Definition~\ref{def:attacks}): two
surfaces---within-boundary splitting and post-toll safe-default
selection---are closed by Paper~A, and the remaining three are closed
here. Cross-boundary aggregation (Theorem~\ref{thm:boundary-agg}) shows
that under common control no finite re-routing of priced actions can
reduce total toll below the boundary potential applied to total
exposure. Interface-compliance adjudication
(Theorem~\ref{thm:interface}) shows that an invalid tool call is a
contract-relevant, not safety-relevant, event, and that the two
principled adjudication rules induce opposite deployment-time
incentives---a distinction we validate quantitatively on committed
cross-model runs in Remark~\ref{rem:t2-validation}. Strategy-proof
model-identity reporting (Theorem~\ref{thm:type}) gives a menu of
identity-indexed safe defaults under which truthful reporting of the
deployed model is weakly dominant, with the componentwise-minimum
penalty schedule given by a Myerson-style ratio. Bundling these with the
Paper~A clauses (Theorem~\ref{thm:composition}) yields joint incentive
compatibility over the entire attack space, and the operator
participation constraint is discharged constructively by a two-parameter
premium family that is simultaneously individually rational and weakly
budget-balanced at the truthful equilibrium (Section~\ref{sec:ir-bb}).
Together these results constitute a first incentive-compatibility
characterisation of the actuarial runtime contract, with participation
and budget-balance constraints satisfied rather than assumed.

The characterisation is deliberately scoped to operator gaming under a
fixed budget, not social-welfare-optimal mechanism design, and it
inherits Paper~A's calibration and boundary assumptions. The remaining
obligations below sharpen these limits into concrete next steps.

\section{Open Obligations}\label{sec:obligations}

\begin{enumerate}[leftmargin=1.5em,itemsep=2pt]
  \item \emph{Non-negative premium feasibility and ex-post penalty rebates.}
  Theorem~\ref{thm:ir-bb} delivers weak budget balance at the truthful
  equilibrium under a two-parameter family of premium schedules
  (Corollary~\ref{cor:ir-bb-params} gives three representative
  points). Two refinements are left for future work. First, the
  toll-rebate parameterisation $(\alpha = 0, \beta = 1)$ produces
  $P_\theta = -\Delta_\theta < 0$ wherever $\Delta_\theta > 0$, i.e.\
  the contract pays the operator at $t = 0$. Many real-world
  contractual norms require $P_\theta \ge 0$; characterising the
  non-negative-premium subregion of the
  $(\alpha, \beta)$-feasibility region is the natural next step.
  Second, formal incorporation of the off-equilibrium penalty-rebate
  mechanism of Remark~\ref{rem:off-eq} requires its own IC analysis
  (the rebate rule must be report-and-execution independent or paid
  only after verified truthfulness, otherwise it perturbs the
  pairwise IC condition of Lemma~\ref{lem:premium-ic}).
  \item \emph{Adversarial conformal calibration.} If the operator can
  choose calibration actions, the envelope quantile $\hat q$ in
  Assumption~16 of Paper~A is exposed to a sixth attack surface not
  covered by the five-attack characterisation of
  Definition~\ref{def:attacks}, hence not closed by
  Theorem~\ref{thm:composition} (see Remark~\ref{rem:coverage}).
  Defence likely requires a contractually fixed, separately-sampled
  audit calibration pool.
  \item \emph{Welfare characterisation.} Theorem~\ref{thm:composition}
  establishes joint IC for one specific parameterisation
  (conservative envelopes, $\xi_{\mathrm{escalate}}$, premium-aware
  minimum $\widetilde{\kappa}_r^\star$). Characterising the
  IC-constrained welfare frontier---how operator welfare varies with
  $\kappa_{\mathrm{esc}}$, the choice between $\xi_{\mathrm{safe}}$
  and $\xi_{\mathrm{escalate}}$, and slack in
  $\kappa_r \ge \widetilde{\kappa}_r^\star$---is left for future work
  (cf.\ Remark~\ref{rem:welfare}).
  \item \emph{Empirical demonstration of Theorem~\ref{thm:interface}.}
  Remark~\ref{rem:t2-validation} discharges this obligation at the
  measured endpoints: it instantiates the operator-utility comparison on
  the committed cross-model runs (measured $f$ and gate-computed reserve
  $\mu_c$) and pins the part~(b) escalation-fee threshold
  $\kappa_{\mathrm{esc}}^\star$. What remains is a full runtime
  intervention---implementing $\xi_{\mathrm{safe}}$ and
  $\xi_{\mathrm{escalate}}$ in the gate and replaying at several
  $\kappa_{\mathrm{esc}}$ values---to trace the incentive endogenously
  rather than analytically.
\end{enumerate}

\bibliographystyle{plainnat}
\begingroup
\raggedright
\nocite{*}
\bibliography{references}
\endgroup

\end{document}